\documentclass[11pt]{article}
\usepackage[margin=1in]{geometry}
\usepackage[utf8]{inputenc}
\usepackage{amsmath,amssymb,amsthm}
\usepackage{hyperref}
\hypersetup{colorlinks,allcolors=[rgb]{0,0.13672,0.95703}}
\usepackage{cleveref,xcolor}
\usepackage[affil-sl]{authblk}
\usepackage{graphicx}
\usepackage{cite}
\usepackage{diagbox}
\usepackage{cancel}
\usepackage{algorithm}
\usepackage{algpseudocode}
\usepackage{soul} 

\newcommand\tr{{\operatorname{tr}}}
\newcommand\diag{{\operatorname{diag}}}
\newcommand\E{{\mathbb{E}}}
\newcommand\C{{\mathbb{C}}}

\renewcommand\v[1]{{\boldsymbol{#1}}}

\newcommand{\p}[0]{\partial}

\newcommand{\R}[0]{\mathbb{R}}

\newtheorem{definition}{Definition}[section]
\newtheorem{theorem}{Theorem}[section]

\newtheorem{lemma}{Lemma}[section]

\newtheorem{remark}{Remark}[section]
\theoremstyle{definition}

\theoremstyle{remark}

\numberwithin{equation}{section}

\renewcommand{\b}{\mathbf{b}}
\newcommand{\A}{\mathbf{A}}
\newcommand{\K}{\mathbf{K}}
\newcommand{\B}{\mathbf{B}}
\newcommand{\D}{\mathbf{D}}
\newcommand{\U}{\mathbf{U}}
\newcommand{\V}{\mathbf{V}}

\newcommand{\M}{\mathbf{M}}
\newcommand{\x}{\mathbf{x}}
\newcommand{\y}{\mathbf{y}}
\newcommand{\z}{\mathbf{z}}
\newcommand{\w}{\mathbf{w}}
\newcommand{\e}{\mathbf{e}}
\newcommand{\g}{\mathbf{g}}

\newcommand{\Q}{\mathbf{Q}}

\newcommand{\I}{\mathbf{I}}

\newcommand{\nnz}{\mathrm{nnz}}

\newcommand{\rank}{\mathrm{rank}}

\newcommand{\polylog}{\mathrm{polylog}}
\newcommand{\dpp}{\mathrm{DPP}}
\newcommand{\range}{\mathrm{range}}
\newcommand{\adj}{\mathrm{adj}}
\renewcommand{\det}{\mathrm{det}}

\renewcommand{\u}{\mathbf{u}}
\renewcommand{\v}{\mathbf{v}}

\renewcommand{\H}{\mathbf{H}}
\renewcommand{\p}{\mathbf{p}}
\renewcommand{\L}{\mathbf{L}}

\renewcommand{\C}{\mathbf{C}}
\renewcommand{\P}{\mathbf{P}}

\renewcommand{\c}{\mathbf{c}}

\DeclareMathOperator*{\argmin}{arg\,min}

\title{Solving Dense Linear Systems Faster Than via Preconditioning}

\author[2]{Micha{\l} Derezi\'nski\thanks{University of Michigan (\texttt{derezin@umich.edu})}
\quad and\quad 
Jiaming Yang\thanks{University of Michigan (\texttt{jiamyang@umich.edu})}}

\begin{document}

\maketitle
\begin{abstract}
    We give a stochastic optimization algorithm that solves a dense $n\times n$ real-valued linear system $Ax=b$, returning $\tilde x$ such that $\|A\tilde x-b\|\leq \epsilon\|b\|$ in time:
    \begin{align*}
        \tilde O\big((n^2+nk^{\omega-1})\log1/\epsilon\big),
    \end{align*}
    where $k$ is the number of singular values of $A$ larger than $O(1)$ times its smallest positive singular~value, $\omega < 2.372$ is the matrix multiplication exponent, and $\tilde O$ hides a poly-logarithmic in $n$ factor. When $k=O(n^{1-\theta})$ (namely, $A$ has a flat-tailed spectrum, e.g., due to noisy data or regularization), this improves on both the cost of solving the system directly, as well as on the cost of preconditioning an iterative method such as conjugate gradient. In particular, our algorithm has an $\tilde O(n^2)$ runtime when $k=O(n^{0.729})$. We further adapt this result to sparse positive semidefinite matrices and least squares regression.
    
    Our main algorithm can be viewed as a randomized block coordinate descent method, where the key challenge is simultaneously ensuring good convergence and fast per-iteration time. In our analysis, we use theory of majorization for elementary symmetric polynomials to establish a sharp convergence guarantee when coordinate blocks are sampled using a determinantal point process. We then use a Markov chain coupling argument to show that similar convergence can be attained with a cheaper sampling scheme, and accelerate the block coordinate descent update via matrix sketching.
\end{abstract}
\newpage
\tableofcontents
\newpage

\section{Introduction}
\label{s:intro}

Solving a linear system $\A\x=\b$ is one of the most classical algorithmic tasks, with applications in scientific computing, engineering, data science, optimization, and more. Many classical deterministic algorithms are known for this task, from direct methods such as Gaussian elimination to iterative methods such as Conjugate Gradient. More recently, there has been much success in developing faster algorithms based on randomized techniques such as sampling and sketching for solving linear systems under specific structural conditions, e.g., graph-structured \cite{spielman2014nearly,koutis2012fast,cohen2018solving} and other classes of sparse linear systems \cite{xia2012superfast,gray2006toeplitz}, as well as extremely over/under-determined systems \cite{sarlos2006improved,woodruff2014sketching}. However, the benefits of using randomized algorithms for general dense linear systems remain unclear.

Most of the work on improving the time complexity of direct methods for solving linear systems has focused on the reduction to matrix multiplication via divide-and-conquer approaches. This has led to the complexity of $O(n^\omega)$ for solving an $n\times n$ linear system $\A\x=\b$, where $\omega < 2.371552$ \cite{williams2023new} is the current best known bound on the matrix multiplication exponent, developed over a long line of works \cite{strassen1969gaussian,pan1984multiply,coppersmith1987matrix,williams2012multiplying}.

On the opposite end of the spectrum from direct methods are iterative methods, which, through a sequence of cheap iterations, construct estimates that gradually converge to the solution. These approaches are widely used in practice for solving large linear systems. One of the most famous of those methods is Conjugate Gradient (CG) / Lanczos algorithm \cite{hestenes1952methods}. This method under exact arithmetic recovers the optimal solution in $n$ iterations, each costing $O(n^2)$ time for a dense system. However, CG is more commonly used as an approximation algorithm for well-conditioned systems, taking $O(\log(1/\epsilon))$ iterations to converge $\epsilon$-close when the condition number $\kappa(\A)$ (the ratio between largest and smallest positive singular values) is bounded by a constant. These dramatically opposing views of the complexity of CG can be unified through careful analysis of its dependence on the singular value distribution of $\A$ \cite{sw09}. In particular, CG will take $k+O(\log(1/\epsilon))$ iterations to converge, when $\A$ has at most $k$ singular values larger than $O(1)$ times the smallest positive singular value, or equivalently, when $\kappa(\A-\A_k)=O(1)$, where $\A_k$ is $\A$'s best rank-$k$ approximation. This leads to the overall complexity of $\tilde O(n^2k)$, where we use $\tilde O$ to hide the logarithmic factors.

The above discussion indicates that solving linear systems with a matrix having only $k=o(n)$ large singular values should be substantially easier than in the worst-case. Understanding this phenomenon is motivated by the real-world prevalence of such linear systems: either due to direct regularization (e.g., $\A=\B+\lambda\I$), as in continuous optimization \cite{boyd2004convex}, machine learning \cite{zhang2013divide}, and statistics \cite{dobriban2018high}, or because the input matrix is distorted by isotropic noise (e.g., $\A=\B+\delta\mathbf{N}$, where $\mathbf{N}$ is the noise), which can be caused by measurement error \cite{loh2011high}, data compression \cite{liang2021pruning}, and privacy enforcement \cite{dwork2014algorithmic}. Yet, the $\tilde O(n^2k)$ complexity of CG is highly unsatisfactory, as it only provides an improvement over direct methods for small $k$, and does not lead to a near-linear time algorithm unless $k=\tilde O(1)$. This leads to our central question:
\begin{quote}
    What is the complexity of solving an $n\times n$ linear system with $k$ large singular values?
\end{quote}

Arguably the most popular strategy of accelerating the complexity of linear system solvers such as CG is via preconditioning. For instance, given access to $\A$'s best rank $k$ approximation $\A_k$, we can easily construct a preconditioner matrix $\M$ that allows CG to solve our problem in $O(\log(1/\epsilon))$ iterations (since $\A_k$ describes
the entire ill-conditioned part of $\A$'s spectrum). While computing $\A_k$ exactly is as expensive as directly solving the linear system, much recent work has been dedicated to fast randomized algorithms for low-rank approximation via matrix sketching. In particular, we can use power iteration \cite{halko2011finding} (or Krylov subspace iteration \cite{musco2015randomized}) to construct a low-rank approximation $\tilde\A_k$ with a spectral norm guarantee $\|\tilde\A_k-\A\|_2 = O(\sigma_{k+1}(\A))$, sufficient to construct a good preconditioner for our problem. These algorithms require (repeatedly) multiplying matrix $\A$ with an $n\times k$ random matrix, leading to the runtime of $\tilde O(n^2k^{\omega-2})$. This complexity can be further refined using fast rectangular matrix multiplication \cite{le2012faster} by replacing the exponent $\omega$ with $\omega(1,1,\log_n(k))\leq \omega$. Nevertheless, for large enough $k$, the resulting preconditioning algorithm will not run in near-linear time.\footnote{Using current matrix multiplication algorithms, we have $\omega(1,1,\alpha)>2$ for any $\alpha> 0.321335$ \cite{williams2023new}. For simplicity, from now on we will present all time complexities using just the square matrix multiplication exponent~$\omega$.}
Some randomized low-rank approximation algorithms avoid power iteration and run in $\tilde O(n^2 + nk^{\omega-1})$ time \cite{clarkson2013low,cohen2015dimensionality,cohen2017input}, but they only guarantee a bound on the Frobenius norm error $\|\tilde\A_k-\A\|_F$ which is worse by up to a $\sqrt n$ factor than the spectral norm bound and thus does not in general lead to a good preconditioner. 

In summary, existing approaches suffer from an undesirable trade-off between the cost of preconditioning and the cost of solving the system, where one of those costs has to be larger than the input size by at least a $\mathrm{poly}(k)$ factor.

\begin{table}
  \centering\begin{tabular}{r||c|c|c}
Problem &Condition on $\A$ 
& Our results& Theorem\\
  \hline\hline
    Dense linear system
        & $\A\in\R^{n\times n}$
& $\tilde O(n^2+nk^{\omega-1})$& Thm \ref{thm:main_2}\\
Least squares & $\A\in\R^{n\times d}$
&$\tilde O(\nnz(\A)+d^2+dk^{\omega-1})$& Thm \ref{thm:main_ls}\\
PSD linear system & $\A\in\mathcal S_n$
&$\tilde O(\nnz^*(\A)+nk^{\omega-1})$& Thm \ref{thm:main_psd}
\end{tabular}
\caption{Summary of our complexity bounds for solving 
linear systems with at most $k$ large singular values to high
precision, hiding $\polylog(n/\epsilon)$ factors, as a function of the square matrix multiplication exponent $\omega$. In the least
squares setting, we consider $\A$ to be a tall matrix, i.e., $n\gg d$, and
in the psd setting, we define $\nnz^*(\A)=n\cdot \max_i\{\nnz(\A_{i,:})\}$.}\label{tab:summary}
\end{table}

\subsection{Main Results} 
In this paper, we avoid the undesirable trade-off between the cost of preconditioning and the cost of solving a dense linear system with $k$ large singular values, by abandoning the preconditioning step in favor of stochastic optimization. 
\begin{theorem}[Dense linear system, simplified Theorem \ref{thm:main_2}]\label{t:main-simple}
    Given $\A\in\R^{n\times n}$, $\b\in\R^n$, $\epsilon>0$ and a constant $C=O(1)$, we can compute $\tilde\x$ such that $\|\A\tilde\x-\b\|\leq\epsilon\|\b\|$ in time:
    \begin{align*}
        \tilde O\Big((n^2+nk^{\omega-1})\log1/\epsilon\Big),
    \end{align*}
    where $k$ is the number of singular values of $\A$ larger than $C$ times its smallest positive singular~value.
\end{theorem}
\begin{remark}
    We note that $\tilde O$ hides only $\polylog(n)$ factors, and the algorithm can be implemented without the knowledge of $k$. Theorem \ref{thm:main_2} further extends to rectangular and inconsistent systems.
\end{remark}
The above theorem implies that we can solve an $n\times n$ linear system with $k$ large singular values in near-linear time, $\tilde O(n^2)$, as long as $k=O(n^{\frac1{\omega-1}})=O(n^{0.729})$.
Moreover, when we compare the time complexity $\tilde O(n^2+nk^{\omega-1})$ against the cost of CG with low-rank preconditioning, the proposed algorithm entirely avoids the $\tilde O(n^2k^{\omega-2})$ cost of power iteration, without sacrificing the complexity of the solver (up to logarithmic factors). See Table \ref{tab:summary} for a summary of our results.

\paragraph{Least squares.} 
Our approach naturally extends to both over- and under-determined systems. In particular, Theorem \ref{thm:main_2} directly shows that we can solve a tall $n\times d$ least squares problem $\min_\x\|\A\x-\b\|$ with $k$ large singular values, in time $\tilde O(nd+dk^{\omega-1})$. This can be further improved for sparse matrices, by combining our algorithm with existing sketching-based preconditioning techniques for tall least squares. Below, we use $\nnz(\A)$ to denote the number of non-zeros in $\A$.
\begin{theorem}[Least squares, simplified Theorem \ref{thm:main_ls}]\label{t:ls-simple}
    Given $\A\in\R^{n\times d}$ with $k$ large singular values, $\b\in\R^n$ and $\epsilon>0$, we can compute $\tilde \x$ so that $\|\A\tilde\x-\b\|^2\leq \min_\x\|\A\x-\b\|^2+\epsilon\|\b\|^2$ in time 
    $$\tilde O\Big((\nnz(\A) + d^2 + dk^{\omega-1})\log1/\epsilon\Big).$$
\end{theorem}

\paragraph{Positive semidefinite systems.}
An important application of our results is solving symmetric positive semidefinite (psd) linear systems, which arise from data covariances, kernel matrices, Hessian matrices, and more. While Theorem \ref{t:main-simple} already provides improved time complexity guarantees for dense psd systems, a careful adaptation of our algorithm allows it to exploit the sparsity structure of the problem, as well as improving the $\polylog(n)$ factors in the bounds. Similarly to prior stochastic psd linear system solvers \cite{lee2013efficient}, our complexity bounds depend on the following notion of sparsity, $\nnz^*(\A)=n\cdot\max_i\{\nnz(\A_{i,:})\}\leq n^2$, which is proportional to $\nnz(\A)$ when the non-zeros are distributed somewhat uniformly. For a psd matrix $\A$, we define $\|x\|_\A=\sqrt{\x^\top\A\x}$.
\begin{theorem}[PSD linear system, simplified Theorem~\ref{thm:main_psd}]\label{t:psd-simple}
Given a psd matrix $\A\in\R^{n\times n}$ with $k$ large eigenvalues, $\b\in\R^n$, and $\epsilon>0$, we can find $\tilde \x$ such that $\|\tilde\x-\A^\dagger\b\|_\A\leq \epsilon \|\A^\dagger\b\|_\A$ in time:
\begin{align*}
\tilde O\Big((\nnz^*(\A)+nk^{\omega-1})\log1/\epsilon\Big).
\end{align*}
\end{theorem}

Our results naturally apply to regularized linear systems, such as kernel ridge regression, i.e., $(\K+\lambda\I)\x=\b$ for a psd matrix $\K$, and $l_2$-regularized least squares, $\min_\x\|\A\x-\b\|^2+\lambda\|\x\|^2$ (as special cases of Theorem~\ref{t:psd-simple} and Theorem~\ref{t:ls-simple}, respectively). In these cases, we provide the first complexity bounds that scale with the number of eigenvalues (of $\K$ or of $\A^\top\A$) larger than the regularization parameter $\lambda$. This is in contrast to prior works that directly exploit this explicit regularization structure \cite{acw17,ftu21}, where the time complexity scales with a notion of $\lambda$-effective dimension which can be arbitrarily larger than our $k$. Thus, remarkably, we achieve a direct runtime improvement over explicitly regularized linear system solvers, even though our algorithms are entirely oblivious to the regularization structure themselves (see Section~\ref{s:related-work} for discussion). 

\paragraph{Numerical stability of our methods.} For the sake of simplicity, all our results assume the exact arithmetic (Real-RAM) model. However, we do not foresee any obstacles in establishing the numerical stability of our algorithms in a Word-RAM model with word sizes polylogarithmic in $n$, the condition number of $\A$, and precision $1/\epsilon$. In fact, our convergence analysis of the main algorithmic framework already allows for approximate computations (Section \ref{sec:approx_sketch_proj}), so that an inner solver can be employed to perform them. Thus, the question reduces to the numerical stability of the inner solver. For this, we use preconditioned conjugate gradient (PCG), where the preconditioner is constructed by sketching a tall matrix (i.e., multiplying by a sparse oblivious subspace embedding), and then computing a singular value decomposition of the sketch. While the conjugate gradient method can behave differently in finite precision than in the Real-RAM model, it has been established that it still obtains essentially the same worst-case convergence rate (stated in our \Cref{lem:cg}) in finite precision \cite{Greenbaum:1989,musco2018stability}.
The preconditioner can be constructed in a stable way using black-box (fast) linear algebra routines \cite{demmel2007fast}.

\subsection{Our Techniques}

A core inspiration for our algorithms is  a broad family of stochastic iterative methods for solving linear systems, which includes randomized Kaczmarz \cite{strohmer2009randomized}, coordinate descent \cite{leventhal2010randomized}, sketch-and-project \cite{gower2015randomized}, and many others \cite{lee2013efficient}. Unlike CG-type methods which do a pass over the entire matrix $\A$ in each iteration, these approaches randomly sub-sample a small number of rows/columns of $\A$ (or construct a small sketch of $\A$) and use this stochastic information to perform a sublinear iterative update. There has been much success in characterizing the convergence and complexity of stochastic methods in the ultra-small sample regime (i.e., single row/column per iteration \cite{strohmer2009randomized}), but the theoretical analysis of the larger sample regime has proven much more elusive. While there has been some progress in this direction over the last few years, it has come with trade-offs: meaningful characterizations of the convergence behavior have required relying on either constructing sophisticated and expensive sub-sampling distributions (such as determinantal point processes \cite{rk20,mutny2020convergence}), or even more expensive dense sketching methods (such as the Gaussian sketch \cite{rebrova2021block,dr22}), rendering the entire procedure computationally impractical.

In light of this context, when designing and analyzing our stochastic iterative solvers, we encountered three main bottlenecks which had to be overcome to establish any improvement over CG-based complexity guarantees for solving dense linear systems:
\begin{enumerate}
    \item Obtaining a sharp convergence guarantee that depends on the spectral structure of $\A$;
    \item Relying on a fast sub-sampling/sketching scheme;
    \item Ensuring fast computation of the iterative updates.
\end{enumerate}

\textbf{1. Convergence guarantee.} As a starting point in our analysis we used recent works on analyzing the convergence rate of block coordinate descent algorithms using Determinantal Point Processes (DPP) as the sub-sampling method. DPPs form a family of discrete distributions over subsets of a finite domain, e.g., $\{1,...,n\}$, where the probabilities of each subset are proportional to a determinant of a certain sub-matrix (the uniform distribution over spanning trees of a graph is one example). Exploiting certain formulas over sums of determinants and adjugate matrices, recent works \cite{rk20,mutny2020convergence} have characterized the convergence rate of block coordinate descent algorithms with DPP and k-DPP (DPP with size $k$ subsets) sub-sampling  in terms of the elementary symmetric polynomials (ESP) of the singular values of matrix $\A$. Of particular relevance to our results are the guarantees obtained by Rodomanov and Kropotov \cite{rk20}. However, their convergence bounds by themselves are not sufficiently sharp to obtain any improvement over CG-type methods, even if we entirely ignored the remaining two bottlenecks. Thus, as a first step, we improve the convergence guarantees of \cite{rk20} for k-DPP-based block coordinate descent by applying theory of majorization, via Schur-concavity of the ratios of ESPs, leading to a factor $k$ improvement in the convergence complexity. In particular, we show that for matrices with $k$ large singular values, using $\tilde O(k)$-sized blocks, we need $\tilde O(n/k)$ iterations to solve a linear system.

\textbf{2. Fast sub-sampling.} There has been much recent work on efficiently sampling from discrete distributions, including strongly Rayleigh measures \cite{anari2016monte} such as k-DPPs. Most recently, Anari, Liu and Vuong \cite{alv22} provided nearly-tight complexity analysis for a Markov chain Monte Carlo (MCMC) algorithm that samples from a k-DPP (to within a small total variation distance). However, the time complexity of their algorithm is still not sufficient to produce a stochastic method that improves upon CG-type solvers for dense linear systems, as it requires performing multiple singular value decompositions of $n \times k$ matrices to produce a single k-DPP sample, at the cost of at least $\tilde O(nk^{\omega-1})$ per iteration. To address this, in the next step of our analysis, we show that the convergence guarantees of k-DPP sampling can be retained (up to logarithmic factors) by a much simpler and cheaper sub-sampling procedure, after preprocessing the linear system with a randomized Hadamard transform \cite{t11}. We show this by coupling a hypothetical run of the MCMC algorithm from \cite{alv22} with our sub-sampling procedure, in such a way that the k-DPP sample will (with high probability) be contained within our sample.

\textbf{3. Fast updates.}  Standard block coordinate methods perform their iterations by reducing the $n\times n$ linear system to a much smaller $k\times n$ linear system, which is extremely under-determined, and then choose a solution with the smallest distance from the current iterate. For $k=O(1)$, this is a very simple and inexpensive update, but as $k$ grows, the cost of solving the system exactly becomes prohibitively expensive. However, we can leverage the fact that, unlike the original system, this new system is extremely rectangular, which makes it much easier to construct an effective preconditioner via sketching \cite{rokhlin2008fast,msm14,woodruff2014sketching}. We prove that $O(\log(n/k))$ (inner) iterations of preconditioned CG can be used to solve this sub-problem to within sufficient accuracy so that the overall convergence complexity of the stochastic method remains unchanged. This improves the cost of the stochastic updates from $O(nk^{\omega-1})$ to $\tilde O(nk+k^\omega)$. Since we showed that $\tilde O(n/k)$ updates are needed to converge, this leads to the final complexity result in Theorem \ref{t:main-simple}.

\subsection{Background and Related Work}
\label{s:related-work}

\textbf{Linear system solvers.}
Extensive literature has been dedicated to solving linear systems both via direct and iterative methods. The direct approaches are based on the strategy of applying a (possibly implicit) representation of $\A^{-1}$ to the vector $\b$. Examples of this approach include Gaussian elimination, Cholesky factorization and QR decomposition. A key breakthrough in the time complexity of direct methods came from Strassen, who showed that matrix inversion is equivalent to matrix multiplication \cite{strassen1969gaussian}. This kickstarted a long line of works on fast $O(n^\omega)$ matrix multiplication \cite{pan1984multiply}, including the Coppersmith-Winograd method \cite{coppersmith1987matrix}. This is still an active area of research \cite{williams2012multiplying,le2012faster,alman2021refined,williams2023new}. 

There has also been much work on accelerating iterative linear system solvers, particularly those based on Krylov subspace methods, such as conjugate gradient / Lanczos. Here, most of the efforts have been concentrated on constructing effective preconditioners, for example via specialized efficient Cholesky factorizations \cite{kyng2016sparsified,kyng2016approximate}, or on improving their numerical stability \cite{peng2021solving,nie2022matrix}. However, these approaches have focused on sparse and/or structured systems, e.g., those decomposable into a sum of simple components, particularly Laplacians of undirected graphs \cite{vaidya1989speeding,spielman2014nearly,koutis2012fast}, and their extensions \cite{cohen2018solving}. Other examples of structured systems with efficient solvers include Hankel/Toeplitz matrices \cite{kailath1979displacement,xia2012superfast} and circulant matrices \cite{gray2006toeplitz}. 

\textbf{Preconditioning regularized systems.}
Linear systems with an explicit regularization term form another class of structured systems that can benefit from effective preconditioning. The main examples here are: the $l_2$-regularized least squares task, $\min_\x\|\A\x-\b\|^2+\lambda\|\x\|^2$; kernel ridge regression, $(\K+\lambda\I)\x=\b$ where $\K$ is a psd kernel matrix; and the regularized Newton step, $(\H+\lambda\I)\x=\g$, where $\H$ is the Hessian matrix, which is symmetric but may not be psd. A long line of works have used randomized sketching, sampling, and the Nystr\"om method \cite{em14,rcr17,mm17} to construct effective low-rank preconditioners for these systems \cite{acw17,ftu21}, relying on the explicit regularizer to avoid power iteration. To construct a good preconditioner for, say, the matrix $\K+\lambda\I$, those algorithms construct a  decomposition of $\K$ with rank at least proportional to the so-called $\lambda$-effective dimension: $d_\lambda=\tr(\K(\K+\lambda\I)^{-1})\leq n$ \cite{em14}. The best known cost of this preconditioning step, even exploiting the structural properties of psd matrices, is at least $\tilde O(n d_\lambda^{\omega-1})$ \cite{mm17}. Similarly as our parameter $k$, the $\lambda$-effective dimension exploits the flat tail of the spectrum of $\K+\lambda\I$. In fact, it is easy to see that such a matrix may have no more than $k=2d_\lambda$ eigenvalues larger than twice its smallest eigenvalue. However the bound cannot be reversed: for any $n$, $k$ and $\lambda$, one can construct an $n\times n$ psd matrix $\K$ such that $\K+\lambda\I$ has only $k$ eigenvalues larger than twice its smallest eigenvalue, and yet its $\lambda$-effective dimension is at least $n/2$, in which case, the preconditioning cost becomes $\Theta(n^\omega)$, as large as the cost of directly solving the system. Thus, our algorithms not only recover but also improve upon the time complexity of these specialized solvers, even without explicit access to the regularization term.

\textbf{Stochastic solvers.} The literature on stochastic iterative methods for solving linear systems has its origins in the Kaczmarz method \cite{kaczmarz37}, which performs iterative updates while cycling through the rows of the system and projecting the current iterate onto the hyperplane formed by an individual row. A randomized version of this method was introduced and analyzed by Strohmer and Vershynin \cite{strohmer2009randomized}, followed up by closely related randomized coordinate descent methods for psd linear systems and least squares by Leventhal and Lewis \cite{leventhal2010randomized}. This triggered much research into stochastic iterative solvers, leading to more general classes of methods that sample blocks of rows instead of a single row \cite{needell2013two,needell2014paved,mutny2020convergence}, or use matrix sketching \cite{rebrova2021block,dr22} via the Sketch-and-Project framework \cite{gower2015randomized}. Despite these efforts, it is still not fully understood when stochastic solvers are faster than deterministic (e.g., Krylov-based) iterative methods. One such result was obtained by Lee and Sidford \cite{lee2013efficient}, who showed that the convergence rate of an accelerated stochastic coordinate descent solver exhibits better condition number dependence than conjugate gradient for psd linear systems. However, this still leads to poor performance for matrices that have an ill-conditioned top-$k$ part of the spectrum, whereas our algorithms avoid this by using blocks of coordinates (rows or columns) as opposed to one coordinate at a time. 

\textbf{Matrix sketching and sub-sampling.} Our techniques are closely related to matrix sketching and sampling techniques commonly used in randomized linear algebra \cite{halko2011finding,woodruff2014sketching,drineas2016randnla,martinsson2020randomized,randlapack_book}, which is an area that has led to improved randomized approximation algorithms for tasks such as least squares regression \cite{sarlos2006improved,rokhlin2008fast,clarkson2013low}, low-rank approximation \cite{cohen2015dimensionality,cohen2017input,li2020input}, $l_p$ regression \cite{meng2013low,cohen2015lp,cd21,wang2022tight}, linear programming \cite{cohen2021solving} and more \cite{song2019relative,jiang2020faster}. Some of the key matrix sketching techniques we use in this work are subspace embeddings \cite{clarkson2013low,nelson2013osnap,meng2013low,c16,chenakkod2023optimal}, which can be used to construct effective preconditioners for highly over- or under-determined linear systems (e.g., least squares). In particular, our fast least squares solver is an accelerated version of the sketch-to-precondition algorithm by Rokhlin and Tygert \cite{rokhlin2008fast}, which, using current subspace embedding techniques, can solve this task in $\tilde O(\nnz(\A) + d^\omega)$ time by sketching the tall $n\times d$ matrix $\A$ down to an $\tilde O(d)\times d$ size, and then performing a QR/SVD decomposition to produce a preconditioner. When the matrix $\A$ has only $k$ large singular values, then this can be improved to $\tilde O(\nnz(\A) +d^2k^{\omega-2})$, by using randomized low-rank SVD with power iteration \cite{halko2011finding}. We instead use our stochastic solver to avoid constructing an explicit preconditioner altogether (see Table \ref{tab:summary} for comparison). Finally, as mentioned earlier, many randomized low-rank approximation methods mentioned above, such as \cite{cohen2015dimensionality,cohen2017input}, avoid power iteration and run in linear time. However, these approaches provide a Frobenius norm error guarantee (as opposed to an approximation of all top $k$ singular values), which does not translate into a strong preconditioner for a linear system.

One of the key matrix sub-sampling techniques we use (albeit mainly for our theoretical analysis, rather than as an algorithmic routine) is called determinantal point processes (DPP), also known as volume sampling. This is a family of discrete distributions which has its origins in physics \cite{macchi1975coincidence} and graph theory \cite{guenoche1983random}. More recently, DPPs have been used in randomized linear algebra \cite{dm21} to construct optimal coresets for the column subset selection problem \cite{deshpande2006matrix,gs12,dkm20} and least squares regression \cite{derezinski2019minimax}, as well as for discrete and stochastic optimization \cite{nikolov2019proportional,derezinski2020bayesian,rk20,mutny2020convergence}. Early algorithmic approaches to DPP sampling required an exact SVD of the input matrix \cite{hough2006determinantal,kt12}, but a more recent line of works showed that this can be avoided \cite{derezinski2019fast,dcv19}, culminating in fast Markov chain Monte Carlo sampling algorithms for DPPs \cite{anari2016monte,anari2020isotropy,alv22}. Still, these methods lead to a substantial overhead in time complexity for our algorithms, which is why we do not use them explicitly except for certain special cases (namely, psd linear systems).

\section{Preliminaries}\label{sec:prelim}
\paragraph{Problem Setup and Notation.}
Throughout most of the paper, we will focus on a linear system $\A\x = \b$ where $\A \in \R^{m \times n}$ and $\b \in \R^m$. Let $\rank(\A) = r$ and $\{\sigma_i\}_{i=1}^r$ be the decreasing singular values of $\A$.  There are two cases we consider. In the first case where $\b \in \range(\A)$, the linear system is consistent and has one unique min-length solution $\x^* = \A^\dagger \b$. In the second case where $\b \notin \range(\A)$, the linear system is inconsistent and there is no exact solution, and we denote $\x^* = \A^\dagger \b$ as the solution of the least squares problem $\min_{\x} \|\A\x - \b\|$.
As a special case, we also consider solving a symmetric positive semidefinite definite (psd) linear system $\A\x = \b$ where $\A\in\mathcal S_n$, $\b\in\R^n$, and we use $\mathcal S_n$ to denote the psd cone.

For vector $\v$, we use $\|\v\|$ to denote its Euclidean norm, and for a psd matrix $\M$, we denote $\|\v\|_\M := \sqrt{\v^\top\M\v}$. For matrix $\A$, we denote $\|\A\|_2$ and $\|\A\|_F$ as its operator norm and Frobenius norm respectively, and denote $\kappa(\A)$ as its condition number. We also define the averaged condition number of $\A$ as $\bar{\kappa}(\A) := \frac{\|\A\|_F \|\A^\dagger\|_2}{\sqrt{\rank(\A)}}\leq \kappa(\A)$, which arises in the analysis of our algorithms. For $1\leq k\leq \rank(\A)$, we denote the best rank-$k$ approximation of $\A$ as $\A_k := \argmin_{\rank(\hat{\A}) = k} \|\A-\hat{\A}\|_2$, and denote $\bar{\kappa}_k := \bar{\kappa}(\A-\A_k)$ (i.e., the averaged condition number of the tail of the spectrum). Let $\{\sigma_i\}_{i=1}^r$ be the singular values of $\A$ in decreasing order, if we assume that there exists $k$ such that $\sigma_i/\sigma_r = O(1)$ holds for all $i\geq k$, then we have $\bar{\kappa}_k = (\frac{1}{r-k}\sum_{j>k}^r \sigma_j^2 / \sigma_r^2)^{1/2} = O(1)$. For simplicity of presentation, throughout the paper we assume that $\bar{\kappa}_k< m^4$, which is a rather loose condition. For two $n \times n$ psd matrices $\A$ and $\B$, we say $\A \approx_{1+\epsilon} \B$ when $\frac{1}{1+\epsilon} \A \preceq \B \preceq (1+\epsilon)\A$ where $\preceq$ denotes the matrix ordering.

For $1 \leq \tau \leq m$, we use $\tbinom{[m]}{\tau}$ to denote the set of all the $\tau$-element subsets of $[m] = \{1, 2, \ldots, m\}$. For each $S\in\tbinom{[m]}{\tau}$, we use $\I_S\in\R^{\tau \times m}$ to denote the matrix obtained from $m\times m$ identity matrix $\I$ by only retaining the rows whose indices are in $S$. For matrix $\L \in \R^{m \times m}$ and subset $S\in\tbinom{[m]}{\tau}$, we use $\L_{S, S} = \I_S\L\I_S^\top  \in \R^{\tau \times \tau}$ to denote the principal submatrix located at the intersection of the rows and columns with indices from $S$, we also denote $\L_{:,S} = \L\I_S^\top$ and $\L_{S,:} = \I_S \L$ respectively. For a square $n\times n$ matrix $\L$, we use $\adj(\L)$ to denote its adjugate matrix, which is an $n\times n$ matrix such that $\adj(\L)_{ij} = (-1)^{i+j}\det(\L_{[n]\backslash j,[n]\backslash i})$.

\paragraph{Elementary symmetric polynomials.} We next introduce some concepts that are needed for the convergence analysis of our algorithms, starting with elementary symmetric polynomials (ESP) and their properties.
\begin{definition}[Elementary symmetric polynomial]
For a given vector $\lambda = (\lambda_1, \ldots, \lambda_m) \in \R^m$, we define its $\ell^{th}$ elementary symmetric polynomial as 
\begin{align*}
\sigma_{\ell}(\lambda) := \sum_{S \in \tbinom{[m]}{\ell}}\prod_{i \in S} \lambda_i.
\end{align*}
\end{definition}

\begin{lemma}[Sum of principal minors]\label{lem:sum_minor}
Let $\L\in\R^{m \times m}$ be a symmetric matrix with eigenvalues $\lambda = (\lambda_1, \ldots, \lambda_m)$ where $\lambda_1 \geq \cdots \geq \lambda_m$. Let $1\leq \ell \leq m$, then we have
\begin{align*}
\sum_{S \in \tbinom{[m]}{\ell}} \det(\L_{S, S}) = \sigma_{\ell} (\lambda).
\end{align*}
\end{lemma}

\paragraph{Determinantal point processes.}
Given a PSD matrix with dimension $m$, a determinantal point process ($\dpp$) is a probability measure defined over $2^{[m]}$, i.e., all the subsets of $[m]$. There is a special class of $\dpp$s where we restrict the size of the subsets to be fixed, we call this a fixed-size $\dpp$ or $k$-DPP, where $k$ is the size of the subsets. Formally we have the following definitions.
\begin{definition}[$k$-DPP]
Given a PSD matrix $\L\in\R^{m \times m}$, a $k$-$\dpp$ is a distribution on $\tbinom{[m]}{k}$ such that $\Pr\{S\} \propto \det(\L_{S,S})$ where $S$ satisfies $|S| = k$. We denote it as $k$-$\dpp(\L)$.
\end{definition}
An important special case of $k$-DPPs, used in some of the DPP sampling algorithms, is known as Projection DPPs.  
\begin{definition}[Projection DPP]
Given a PSD matrix $\L\in\R^{m \times m}$, a $k$-$\dpp$ is called a Projection $\dpp$ iff $k = \rank(\L)$. We denote it as $\mathrm{P}$-$\dpp(\L)$.
\end{definition}
When the PSD matrix $\L$ is clear in the text, we omit $\L$ and denote the $k$-DPP and projection DPP as $k$-DPP and P-DPP respectively. We have the following lemma which states that a $k$-DPP can be decomposed into a mixture of projection DPPs. This result has been used to develop some of the first DPP sampling algorithms, which require an eigendecomposition of~$\L$.
\begin{lemma}[Algorithm 8 and Theorem 5.2 in \cite{kt12}]\label{lem:sample_dpp_projection}
Given $\L \in \R^{m \times m}$ with $\rank(\L) = r$, let $\L = \sum_i \lambda_i \u_i\u_i^\top$ be the eigendecomposition of $\L$ where $\lambda_1\geq \cdots\geq \lambda_m$. Suppose that we sample $\gamma = (\gamma_1, \ldots, \gamma_r)$ where $\gamma_i\sim\mathrm{Bernoulli}(\frac{\lambda_i}{\lambda_i+1})$ are independent, and then sample $S \sim \mathrm{P}$-$\dpp(\sum_i \gamma_i \u_i\u_i^\top)$. Then, conditioned on the event that $\sum_{i} \gamma_i = k$, we have $S \overset{d}{=} S_{\dpp}$, where $S_{\dpp} \sim k$-$\dpp(\L)$.
\end{lemma}

\paragraph{Markov chain $k$-DPP sampling.}
Recent works have shown that we can approximately sample from a $k$-DPP without an eigendecomposition of the $\L$ matrix, by relying on Markov chain Monte Carlo (MCMC) techniques, which apply more generally to all strongly Rayleigh measures \cite{alv22}. 
\begin{lemma}[Corollary 7 in \cite{alv22}]
\label{lem:anari_main_1}
Given an $m\times m$ PSD matrix $\L$, there is an algorithm that outputs $s$ independent approximate samples from $k$-$\dpp(\L)$ in time $\tilde{O}(mk^{\omega-1} + sk^\omega)$.
\end{lemma}
The above result is based on the following down-up walk sampling algorithm (Algorithm~\ref{alg:dpp_sample}). Given a discrete measure $\mu$ over $\tbinom{[m]}{\tau}$, this procedure will converge to $\mu$ as its stationary distribution. 
\begin{algorithm}[!ht]
\caption{Down-up walk on the complement distribution for a discrete measure $\mu$.}
\label{alg:dpp_sample}
\begin{algorithmic}[1]
\For{$i = 0, 1, 2, \ldots$}
\State From all $t$-sized supersets of $S_i$, select one uniformly at random and name it $T_i$;
\State Select among $k$-sized subsets of $T_i$ a random set $S_{i+1}$ with $\Pr\{S_{i+1}\}\propto \mu(S_i)$.
\EndFor
\end{algorithmic}
\end{algorithm}

\paragraph{Total variation distance.}
In the analysis of sampling a $k$-DPP sample, we will need the following definitions and properties of total variation distance.
\begin{definition}[Total variation distance]
Let $\mu, \nu$ be probability measures defined on a measurable space $(\Omega, \mathcal{F})$, the total variation distance between $\mu$ and $\nu$ is defined as
\begin{align*}
\delta(\mu, \nu) = \sup_{A\in\mathcal{F}} |\mu(A) - \nu(A)|.
\end{align*}
\end{definition}

\begin{definition}[Coupling, see Definition 4.12 in \cite{v14}]
Let $\mu, \nu$ be two probability measures defined on $\Omega_1, \Omega_2$ respectively. We say a joint distribution of random variable $Z = (X,Y)$ over $\Omega_1 \times \Omega_2$ is a coupling of $\mu$ and $\nu$, if it has marginal distributions $X\sim \mu$ and $Y\sim \nu$.
\end{definition}

The following standard lemma states that the total variation distance between $\mu$ and $\nu$ is the minimum probability
that random variables $X\sim \mu$ and $Y\sim \nu$ do not coincide.

\begin{lemma}[Total variation, see Example 4.14 in \cite{v14}]
\label{lem:tv_distance}
Let $\mu, \nu$ be probability measures defined on some $(\Omega,\mathcal F)$. Then, we have
$\inf_{(X,Y)}\Pr\{X\neq Y\} = \delta(\mu, \nu)$,
where the infimum is taken over all couplings of $\mu$ and $\nu$.
\end{lemma}

\paragraph{Randomized Hadamard transform.} We use the randomized Hadamard transform (RHT) as a preprocessing step in our algorithms, to ensure that a simple block coordinate sampling procedure works as well as a determinantal point process.
\begin{definition} \label{hada} The Hadamard matrix $\H_{2^k}$ of dimension $2^k \times 2^k$ is obtained using the following recurrence relation
\begin{align*}
    \H_0 &= [1],\qquad
    \H_{2n} = \begin{bmatrix}
        \H_n & \H_n \\
        \H_n & -\H_n
    \end{bmatrix}.
\end{align*}
\end{definition}
\noindent
Throughout the paper, we will drop the subscript from $\H$, when the dimensions are clear.
\begin{definition} \label{def:RHT}
The randomized Hadamard transform (RHT) of an $m \times n$ matrix $\A$ is the product $\frac{1}{\sqrt{m}}\H\D\A$, where $\D$ is a random $m \times m$ diagonal matrix whose entries are independent random $\pm1$ signs. By padding $\A$ with zero rows if necessary, we may assume that $n$ is a power of~2. 
\end{definition}

In the analysis of using randomized Hadamard transform, one key technique we use is the tail bound for a Rademacher random vector, which is defined as $\mathbf{d} = (d_1, \ldots, d_n)$ where each $d_i$ is an independent random sign variable that equals $1$ with probability $1/2$ and $-1$ otherwise. We have the following tail bound lemma.
\begin{lemma}[Rademacher tail bound, e.g., \cite{t11}]\label{lem:rademacher}
Let $f$ be a convex function on vectors with Lipschitz bound $|f(\x) - f(\y)| \leq L \|\x-\y\|$. Let $\mathbf{d}$ be a Rademacher vector. Then, for all $t \geq 0$,
\begin{align*}
\Pr\left\{f(\mathbf{d}) \geq \E[f(\mathbf{d})] + Lt\right\} \leq e^{-t^2/8}.
\end{align*}
\end{lemma}

\paragraph{Conjugate gradient.}
We have the following algorithm and convergence guarantee of using conjugate gradient method to solve positive semidefinite linear system $\A\A^\top \u = \b$, where we are only given $\A$ and $\b$ without knowing $\A\A^\top$.
\begin{lemma}[Convergence of CG, \cite{h18}]
\label{lem:cg}
Consider solving positive semidefinite linear system $\A\A^\top \u = \b$ where $\A\in\R^{\tau \times n}, \b \in\R^{\tau}$ are given. Let $\kappa := \kappa(\A)$ be the condition number and $\u^* = (\A\A^\top)^{\dagger}\b$ be the min-length solution. Then, running conjugate gradient by starting from $\u_0 = \mathbf{0}$ will yield the following convergence result:
\begin{align*}
\|\u_s - \u^*\|_{\A\A^\top} \leq 2 \left(\frac{\kappa - 1}{\kappa+1} \right)^s \|\u^*\|_{\A\A^\top}.
\end{align*}
We denote the output of conjugate gradient as $\u_s = \mathrm{CG}(\A,\b,s)$.
\end{lemma}

\paragraph{Subspace embeddings.}
To construct the preconditioner in our inner iteration for solving the linear system, as well as for obtaining an improved input sparsity time algorithm for least squares, we use the following result which guarantees a $(1+\epsilon)$-subspace embedding with an optimal embedding dimension \cite{chenakkod2023optimal}. We slightly adapt their result for our setting, so that it holds with a high enough probability to apply a union bound across all iterations of our algorithms.
\begin{lemma}[Adapted from Theorem 1.4 in \cite{chenakkod2023optimal}]\label{lem:precondition_sparse}
Given $\A\in\R^{n \times d}, \epsilon < 1/2$, $\delta < 1/2$, in time $O(\nnz(\A)\log(d / \delta) / \epsilon + d^2\log^4(d/\delta)/\epsilon^6)$ we can compute $\mathbf{\Phi}\A$ where $\mathbf{\Phi}\in\R^{\phi \times n}$ is an embedding matrix with $\phi = O((d+\log(1/\delta)) / \epsilon^2)$, such that with probability $1-\delta$ we have
\begin{align*}
\forall \x\in\R^d, ~~\frac{1}{1+\epsilon} \|\A\x\| \leq \|\mathbf{\Phi}\A\x\| \leq (1+\epsilon) \|\A\x\|.
\end{align*}
\end{lemma}

\section{Main Algorithms}\label{sec:result}
In this section we state our main algorithms and main results for solving a dense linear system $\A\x = \b$ where $\A\in\R^{m\times n}$. For the sake of completeness, and because they provide somewhat different guarantees in somewhat different regimes, we propose two algorithms. 

The first algorithm (Algorithm~\ref{alg:main}) can be viewed as a Kaczmarz-type stochastic solver, in that it sub-samples the rows of $\A$ and projects its iterate onto the hyperplane defined by this smaller linear system. As other Kaczmarz-type methods, this algorithm converges to the optimum solution of a consistent linear system, and its convergence rate is naturally measured in the Euclidean distance from the optimum, obtaining a guarantee of the form  $\|\tilde\x-\x^*\|\leq\epsilon\|\x^*\|$.

The second algorithm (Algorithm~\ref{alg:main_2}) can be viewed more as a coordinate descent-type method, which sub-samples the columns of $\A$. This algorithm actually solves an implicit psd linear system $\A^\top\A\x=\c$, but by setting $\c=\A^\top\b$, we can use it to solve the $\A\x=\b$ system via the normal equations. This approach has the advantage that when the original linear system is inconsistent, then the procedure will converge to the least squares solution $\x^*=\min_\x\|\A\x-\b\|=\A^\dagger\b$. Another difference from the Kaczmarz-type algorithm is that the convergence guarantees of Algorithm \ref{alg:main_2} are expressed as $\|\A(\tilde\x-\x^*)\|\leq\epsilon\|\A\x^*\|$.

In order to ensure fast convergence with a cheap sub-sampling scheme, we first transform the linear system to an equivalent form by using the randomized Hadamard transform (RHT). The transformed system enjoys the property that the corresponding determinantal point process is nearly isotropic, which makes it easier to sample, see Section~\ref{sec:sampling}. Note that our two algorithms use slightly different preprocessing steps: the Kaczmarz-type method (Algorithm~\ref{alg:main}) applies the RHT from the left, $\A \leftarrow \frac{1}{\sqrt{m}}\H\D\A$,  because it uses row sampling; whereas the coordinate descent-type method (Algorithm \ref{alg:main_2}) applies the RHT from the right, $\A \leftarrow \frac{1}{\sqrt{n}}\A\D\H$, because it uses column sampling. In either case, this transformation can be done in $\tilde O(mn)$ time, and since the RHT is an orthogonal matrix, we can easily recover the solution to the original system. 

Both algorithms use an inner solver to approximately solve the small linear system obtained via sub-sampling (the projection step). We do this by constructing a preconditioner for the smaller system and use CG as an inner solver. The construction of the preconditioner is based on a subspace embedding  sketch, leveraging the fact that the small system is highly rectangular,  see Algorithm~\ref{alg:preconditioner} and Lemma~\ref{lem:precondition_sparse}.

\subsection{Fast Kaczmarz-type Solver}
Next we state our first main algorithm and convergence result which solves a consistent linear system. Recall that we define $\bar{\kappa}(\A) = \|\A\|_F\|\A^\dagger\|_2 / \sqrt{r} \leq \kappa(\A)$ as the averaged condition number for matrix $\A$, and define $\bar{\kappa}_k(\A) = \bar{\kappa}(\A-\A_k)$ for $k < r$. Denote $\{\sigma_i\}_{i=1}^r$ as the singular values of $\A$ in decreasing order, then we have $\bar{\kappa}_k(\A) = (\frac{1}{r-k}\sum_{j > k}^r \sigma_j^2 / \sigma_r^2)^{1/2} \leq \sigma_{k+1} / \sigma_r$. Moreover, notice that the statement ``with probability $0.99$'' depends only on ensuring that the preprocessing construction $\frac{1}{\sqrt{m}}\H\D\A$ is successful, and comes from Lemma~\ref{lem:approx_iso_fixed} by choosing $\delta = 0.01$.
\begin{algorithm}[!ht]
\caption{Fast Kaczmarz-type solver for consistent linear systems.}
\label{alg:main}
\begin{algorithmic}[1]
\State \textbf{Input: }matrix $\A\in\R^{m\times n}$, vector $\b \in \R^m$, sample size $\tau$, iterate $\x_0$, inner iteration $s_{\max}$, outer iteration $t_{\max}$;
\State Compute $\A \leftarrow \frac{1}{\sqrt{m}}\H\D\A$ and $\b \leftarrow \frac{1}{\sqrt{m}}\H\D\b$; \Comment{Takes $O(mn \log m)$ time.}
\For{$t = 0, 1, \ldots, t_{\max}-1$}
\State Generate $\tilde{S}$ by calling Algorithm~\ref{alg:dpp_sample_real}; \Comment{Takes $O(\tau)$ time.}
\State Compute $\tilde{\A} \leftarrow \I_{\tilde{S}}\A$; \Comment{Subsampling: takes $O(n\tau)$ time.}
\State Compute $\tilde{\b} \leftarrow \I_{\tilde{S}}\A\x_t - \I_{\tilde{S}}\b$; \Comment{Takes $O(n\tau)$ time.}
\State Generate preconditioner $\M$ by calling Algorithm~\ref{alg:preconditioner} with $\tilde{\A}$; \Comment{Takes $\tilde{O}(n\tau + \tau^\omega)$ time.}
\State Compute $\u_{s_{\max}}\leftarrow\mathrm{CG}(\M^\top\tilde{\A}, \M^\top \tilde{\b}, s_{\max})$ with Lemma~\ref{lem:cg}; \Comment{Takes $O(n\tau s_{\max})$ time.}
\State Compute $\w_t \leftarrow \tilde{\A}^\top \M \u_{s_{\max}}$; \Comment{Takes $O(n\tau)$ time.}
\State Update $\x_{t+1} \leftarrow \x_t - \w_t$; \Comment{Takes $O(n)$ time.}
\EndFor \\
\Return $\tilde{\x} = \x_{t_{\max}}$; \Comment{Solves $\A\x=\b$.}
\end{algorithmic}
\end{algorithm}

\begin{theorem}[Fast Kaczmarz-type solver]\label{thm:main}
Given matrix $\A\in\R^{m\times n}$ with $\rank(\A) = r$ and $\b \in \R^m$, let $\{\sigma_i\}_{i=1}^r$ be the singular values of $\A$ in decreasing order. For $\log m \leq k < r$, define $\bar{\kappa}_k := (\frac{1}{r-k}\sum_{j > k}^r \sigma_j^2 / \sigma_r^2)^{1/2}$ and assume for simplicity that $\bar{\kappa}_k \leq m^4$.
Assume the linear system $\A\x = \b$ is consistent (i.e., $\b \in \range(\A)$) and let $\x^* = \A^\dagger\b$ be its min-length solution. Assume $\x_0 - \x^*$ is not in the nullspace of $\A$. Then conditioned on an event that holds with probability $0.99$ (and only depends on $\frac{1}{\sqrt{m}}\H\D\A$), for $C=O(1)$, running Algorithm~\ref{alg:main} with choice $\tau \geq Ck\log^3 m$ and $s_{\max} = O(\log(r\bar{\kappa}_k / \tau))$ will yield the following convergence result:
\begin{align*}
\E\|\x_t - \x^*\|^2 \leq \left(1 - \frac{\tau/r}{10\bar{\kappa}_k^2\log^3 m}\right)^t \cdot \|\x_0 - \x^*\|^2.
\end{align*}
By further choosing $\tau = O(k\cdot \log^3 m), \x_0 = \mathbf{0}$ and $t_{\max} = O(r \bar{\kappa}_k^2/ k \cdot\log(1/\epsilon))$ for given $\epsilon>0$, Algorithm~\ref{alg:main} outputs $\tilde{\x}$ such that $\|\tilde{\x} - \x^*\|^2 \leq \epsilon \|\x^*\|^2$ holds with probability $0.98$ in time:
\begin{align*}
O\Big(mn \log m + (nr\log^4 m + rk^{\omega-1} \log^{3\omega}m)\cdot \bar{\kappa}_k^2\log1/\epsilon\Big).
\end{align*}
\end{theorem}

\subsection{Fast Coordinate Descent-type Solver}
Next, we present our second main algorithm (fast coordinate descent-type solver), and a corresponding convergence result, which solves a linear system that does not necessarily need to be consistent. We note that while this algorithm can be used directly on tall least squares problems, we provide an improved variant for this setting, which can leverage the sparsity of $\A$ to further improve the runtime (see Section \ref{s:ls}).

\begin{theorem}[Fast coordinate descent-type solver]\label{thm:main_2}
Given matrix $\A\in\R^{m\times n}$ with $\rank(\A) = r$ and $\b \in \R^m$, let $\{\sigma_i\}_{i=1}^r$ be the singular values of $\A$ in decreasing order. For $\log n \leq k < r$, define $\bar{\kappa}_k := (\frac{1}{r-k}\sum_{j > k}^r \sigma_j^2 / \sigma_r^2)^{1/2}$ and assume for simplicity that $\bar{\kappa}_k \leq m^4$.
Let $\x^* = \A^\dagger\b$ be the least squares solution of the (possibly inconsistent) linear system $\A\x = \b$. Assume $\x_0 - \x^*$ is not in the nullspace of $\A$. Then conditioned on an event that holds with probability $0.99$ (and only depends on $\frac{1}{\sqrt{n}}\A\D\H$), for $C = O(1)$, running Algorithm~\ref{alg:main_2} with choice $\c = \A^\top\b, \tau \geq Ck \log^3 n$ and $s_{\max} = O(\log(r\bar{\kappa}_k / \tau))$ will yield the following convergence result:
\begin{align*}
\E\|\A(\x_t - \x^*)\|^2 \leq \left(1 - \frac{\tau / r}{10\bar{\kappa}_k^2\log^3 n}\right)^t \cdot \|\A(\x_0 - \x^*)\|^2.
\end{align*}
By further choosing $\tau = O(k\cdot \log^3 n), \x_0 = \mathbf{0}$ and $t_{\max} = O(r\bar{\kappa}_k^2/k \cdot \log(1/\epsilon))$ for given $\epsilon>0$, Algorithm~\ref{alg:main_2} outputs $\tilde{\x}$ such that $\|\A(\tilde{\x} - \x^*)\|^2 \leq \epsilon \|\A\x^*\|^2$ holds with probability $0.98$ in time:
\begin{align*}
O\Big(mn \log n + (mr \log^4 n + rk^{\omega-1} \log^{3\omega}n)\cdot\bar{\kappa}_k^2\log1/\epsilon\Big).
\end{align*}
\end{theorem}
\begin{remark}\label{rem:main}
By further assuming that the linear system $\A\x = \b$ is consistent (i.e., $\b \in \range(\A)$), we know that $\x^* = \A^\dagger \b$ is indeed a solution of the linear system and $\A\x^* = \b$. Suppose we choose $k \geq \log n$ be such that $\frac{\sigma_i}{\sigma_r} = O(1)$ holds for all $i \geq k$, then we have $\bar{\kappa}_k = O(1)$. Then, Algorithm~\ref{alg:main_2} outputs $\tilde{\x}$ such that $\|\A\tilde{\x} - \b\|^2 \leq \epsilon \|\b\|^2$ with probability $0.98$ in time:
\begin{align*}
\tilde O\Big(mn + (m  r + rk^{\omega-1})\log1/\epsilon\Big).
\end{align*}
If we do not know such $k$ in advance, we can use the following doubling trick to find the right~$k$: we start from a guess $\tilde{k} = 1$ and run Algorithm~\ref{alg:main_2} to obtain an output $\tilde{\x}$ and check if we have $\|\A\tilde{\x} - \b\|^2 \leq \epsilon \|\b\|^2$. If this equation holds, then we can stop with the current value of $\tilde{k}$ and the algorithm solves the linear system; if not, we can simply double the value of $\tilde{k}$ and run Algorithm~\ref{alg:main_2} again. Notice that $\tilde{k}$ satisfies the definition of $k$ as long as $\tilde{k} \geq k$, since $\sigma_i / \sigma_r = O(1)$ also holds for all $i \geq \tilde{k}$. In this way, we need to run Algorithm~\ref{alg:main_2} at most $\log r$ times (since $k\leq r$), where each run costs $\tilde O((mn + rk^{\omega-1})\log1/(\epsilon\delta))$ in order to guarantee that the failure probability is at most $\delta$. Thus the final cost is $\tilde O((mn + rk^{\omega-1})\log1/(\epsilon\delta))$ with $\tilde{O}(\cdot)$ hiding $\polylog(n)$ factors.

\end{remark}

\begin{algorithm}[!ht]
\caption{Fast coordinate descent-type solver for linear systems.}
\label{alg:main_2}
\begin{algorithmic}[1]
\State \textbf{Input: }matrix $\A\in\R^{m\times n}$, vector $\c \in \R^n$, uniform sample size $\tau$, iterate $\x_0$, inner iteration $s_{\max}$, outer iteration $t_{\max}$;
\State Compute $\A \leftarrow \frac{1}{\sqrt{n}}\A\D\H$ and $\c \leftarrow \frac{1}{\sqrt{n}}\H\D\c$; \Comment{Takes $O(mn \log n)$ time.}
\State Compute $\y_0 \leftarrow \A\x_0$; \Comment{Invariant: $\y_t = \A\x_t$.}
\For{$t = 0, 1, \ldots, t_{\max}-1$}
\State Generate $\tilde{S}$ by calling Algorithm~\ref{alg:dpp_sample_real}; \Comment{Takes $O(\tau)$ time.}
\State Compute $\tilde{\A} \leftarrow \A\I_{\tilde{S}}^\top$; \Comment{Subsampling: takes $O(m\tau)$ time.}
\State Compute $\tilde{\c} \leftarrow \tilde{\A}^\top\y_t - \I_{\tilde{S}}\c$; \Comment{Takes $O(m\tau)$ time.}
\State Generate preconditioner $\M$ by calling Algorithm~\ref{alg:preconditioner} with $\tilde{\A}^\top$; \Comment{Takes $\tilde{O}(m\tau + \tau^\omega)$ time.}
\State Compute $\w_{s_{\max}}\leftarrow\mathrm{CG}(\M^\top\tilde{\A}^\top, \M^\top \tilde{\c}, s_{\max})$ with Lemma~\ref{lem:cg}; \Comment{Takes $O(m\tau s_{\max})$ time.}
\State Compute $\x_{t+1} \leftarrow \x_t - \I_{\tilde{S}}^\top \w_{s_{\max}}$; \Comment{Takes $O(\tau)$ time.}
\State Compute $\y_{t+1} \leftarrow \y_t - \tilde{\A} \w_{s_{\max}}$; \Comment{Takes $O(m\tau)$ time.}
\EndFor \\
\Return $\tilde{\x} = \frac{1}{\sqrt{n}}\D\H \x_{t_{\max}}$; \Comment{Solves $\A^\top\A\x=\c$.}

\end{algorithmic}
\end{algorithm}

\section{Overview of the Analysis}\label{sec:analysis}
For simplicity, in this section we only provide analysis and proof sketch for Algorithm~\ref{alg:main} and Theorem~\ref{thm:main}, which is about solving a consistent linear system $\A\x = \b$ where $\A\in\R^{m\times n}, \b \in\R^m$. For the analysis and proof of Algorithm~\ref{alg:main_2} and Theorem~\ref{thm:main_2}, which is similar, see Section~\ref{sec:proof_main_2}. 

\subsection{Expected Projection under DPP Sampling} \label{sec:analysis_1}
Let $\I_{S} \in \R^{\tau \times m}$ be a sampling matrix with $\tau \leq n$. Our analysis is based on an iterative method to solve the linear system which we will refer to as \emph{sketch-and-project}, following \cite{gower2015randomized} (it is also known as the block Kaczmarz method):
\begin{align}\label{eq:sketch_project_rec}
\x_{t+1} = \x_t - \A^\top\I_{S}^\top (\I_{S}\A\A^\top\I_{S}^\top)^\dagger(\I_{S}\A\x_t - \I_{S}\b) = \x_t - \underbrace{(\I_{S}\A)^\dagger(\I_{S}\A\x_t - \I_{S}\b)}_{\w_t^*}.
\end{align}
Denote $\x^* = \A^\dagger\b$ as the min-length solution and denote $\P_S = \A^\top\I_{S}^\top (\I_{S}\A\A^\top\I_{S}^\top)^\dagger\I_{S}\A$. We have
\begin{align*}
\x_{t+1} - \x^* = (\I - \P_S)\x_t + \A^\top\I_{S}^\top (\I_{S}\A\A^\top\I_{S}^\top)^\dagger\I_{S}\b - \x^* = (\I - \P_S)(\x_t - \x^*).
\end{align*}
Notice that $\P_S$ is a projection matrix and thus $\P_S^2 = \P_S$. Since $\x_0 - \x^*$ is not in the nullspace of $\E[\P_{S}]$, we know $\x_t-\x^*$ will always not be in the nullspace of $\E[\P_{S}].$\footnote{For the detailed reason see proof of Lemma~\ref{lem:main_convergence}.} Thus we have
\begin{align*}
\E\|\x_{t+1} - \x^*\|^2 = & ~ \E[(\x_t - \x^*)^\top(\I-\P_{S})^2(\x_t - \x^*)] \nonumber \\
= & ~ (\x_t - \x^*)^\top\E[(\I-\P_{S})^2](\x_t - \x^*) \nonumber \\
= & ~ (\x_t - \x^*)^\top(\I - \E[\P_{S}])(\x_t - \x^*) \nonumber \\
\leq & ~ (1 - \lambda_{\min}^+(\E[\P_{S}])) \cdot \|\x_t - \x^*\|^2.
\end{align*}
In order to obtain the convergence rate of sketch-and-project, we need to lower bound the quantity $\lambda_{\min}^+(\E[\P_{S}])$, which turns out to be quite challenging for most sampling schemes. To address this, following recent prior work \cite{rk20}, we will consider $S$ sampled according to $k$-$\dpp(\A\A^\top)$. \cite{rk20} provide a lower bound of this quantity by relating the expectation under a determinantal point process to a formula for the sum of adjugates of sub-matrices, expressed in terms of the elementary symmetric polynomials of the eigenvalues of $\A\A^\top$. Directly applying their results to our setting with $k' = O(k)$, we can obtain a lower bound of order $\Omega(\frac{1}{r})$, where $r$ is the rank of $\A$ and we hide the dependence on $\bar\kappa_k$. In the following lemma, we refine the proof by providing a sharper bound on the expression derived by \cite{rk20}, using Schur concavity of the ratios of elementary symmetric polynomials. Namely, we show that if we use a slightly larger DPP sample size $k'=2k$, then $\lambda_{\min}^+(\E[\P_{S_\dpp}]) =\Omega( \frac{k}{r})$, an improvement by a factor $k$ which ends up being crucial in the complexity analysis (see Section~\ref{sec:convergence}). Going forward, to simplify the notation we will often use $k$ instead of $k'$, since they are within a factor 2 of one another.
\begin{lemma}[Improved bound]\label{lem:improve_bound}
Given matrix $\A\in\R^{m\times n}$ with rank $r$.
Let $S_{\dpp}\sim k'$-$\dpp(\A\A^\top)$ where $k'<r$, let $\I_S = \I_{S_{\dpp}} \in \R^{k' \times m}$ be the corresponding sampling matrix and define the projection matrix $\P_{S_{\dpp}} := \A^\top\I_S^\top (\I_S\A\A^\top\I_S^\top)^\dagger\I_S \A$. For any $k < k'$ we have
\begin{align}\label{eq:proj_lower_bound}
\lambda_{\min}^+(\E[\P_{S_{\dpp}}]) \geq \frac{k'-k}{k'-k-1 + (r-k)\bar{\kappa}_{k}^2}.
\end{align}
\end{lemma}

\subsection{Replacing DPP with Uniform Sampling}\label{sec:analysis_2}

Despite many recent improvements, sampling a $k$-DPP is still too expensive for our problem. Currently the best known algorithm for sampling $k$-DPPs is obtained by \cite{alv22}, which gives an $\tilde{O}(mk^{\omega-1} + sk^\omega)$ complexity to approximately obtain $s$ samples according to a $k$-DPP based on an $m\times m$ PSD matrix, see Lemma~\ref{lem:anari_main_1}. Since in our case we are dealing with a rectangular non-PSD matrix $\A\in\R^{m\times n}$, after adapting their algorithm to this setting, the cost of sampling becomes $\tilde{O}(mnk^{\omega-2} + s n k^{\omega-1})$, see Section~\ref{sec:compare_sample_dpp}. However this is still too expensive. Can we get rid of $\tilde{O}(mnk^{\omega-2})$? Notice that we do not precisely need a $k$-DPP sample, but rather, it suffices to produce a sample that \emph{contains} a $k$-DPP, since the latter can actually guarantee a faster convergence rate, see Lemma~\ref{lem:project_order}. We propose Algorithm~\ref{alg:dpp_sample_real} which only uses \emph{uniform sampling} with replacement, and show that, after preprocessing the linear system with the randomized Hadamard transform, this algorithm will output a sample $\tilde{S}$ that contains a $k$-DPP sample with high probability. The size of such a uniform sample is only $O((k + \log\frac{1}{\delta}) \log^3 m)$, which is almost linear in $k$, see Lemma~\ref{lem:main_coupling} and Section~\ref{sec:sampling}.

By combining Lemma~\ref{lem:improve_bound}, Lemma~\ref{lem:main_coupling} and Lemma~\ref{lem:project_order}, we have the following main result for our sampling algorithm. Notice that we first apply randomized Hadamard transform $\Q= \frac{1}{\sqrt{m}}\H \mathbf{D}$ to matrix $\A$ as a preprocessing step, in order to guarantee that with high probability the marginal probabilities of the corresponding $k$-DPP are approximately isotropic (see Lemma~\ref{lem:approx_iso_fixed}). For the full proof of Lemma~\ref{lem:main_sampling} see Section~\ref{sec:proof_sample}.

\begin{lemma}[Sampling]\label{lem:main_sampling}
Given matrix $\A\in\R^{m\times n}$ with rank $r$,
denote $\mathbf{Q} = \frac{1}{\sqrt{m}}\H \mathbf{D}$ as the RHT matrix and $\bar{\A} = \Q\A$ as the preprocessed matrix. Conditioned on an event that holds with probability $0.99$ (and only depends on RHT), for $k \geq \log m$, if we draw $\tau$ uniform samples from $[m]$ (according to Algorithm~\ref{alg:dpp_sample_real}) such that $\tau \geq Ck\log^3 m$ for some constant $C = O(1)$, then the resulting set $\tilde S\subseteq[m]$ satisfies
\begin{align*}
\lambda_{\min}^+(\E[\P_{\tilde{S}}]) \geq \frac{\tau / r}{2.5\bar{\kappa}_k^2\log^3 m}
\end{align*}
where $\P_{\tilde{S}} := \bar{\A}^\top \I_{\tilde{S}}^\top(\I_{\tilde{S}}\bar{\A}\bar{\A}^\top\I_{\tilde{S}}^\top)^\dagger\I_{\tilde{S}}\bar{\A}$ is the projection matrix and $\I_{\tilde{S}} \in \R^{\tau \times m}$ is the row sampling matrix according to $\tilde{S}$.
\end{lemma}

\subsection{Approximate Sketch-and-Project}\label{sec:analysis_3}
In the update rule of sketch-and-project (see Eq.\eqref{eq:sketch_project_rec}), the remaining problem is that we need to compute $\w_t^* := (\I_{\tilde{S}}\bar{\A})^\dagger(\I_{\tilde{S}}\bar{\A}\x_t - \I_{\tilde{S}}\bar{\b})$ in each iteration, where $\bar\A=\Q\A$ and $\bar\b=\Q\b$. Directly computing this quantity takes $O(n\tau^{\omega-1})$. Instead of directly computing $\w_t^*$, we show that it is possible to \emph{approximate} this quantity in $\tilde{O}(n\tau+\tau^\omega)$ time (see Section~\ref{sec:approx_sketch_proj}). Denote $\tilde{\A} = \I_{\tilde{S}}\bar{\A} \in \R^{\tau \times n}$ and $\tilde{\b}_t = \I_{\tilde{S}}\bar{\A}\x_t - \I_{\tilde{S}}\bar{\b}$, let $\M \in \R^{\tau \times \tau}$ be a preconditioner such that $\tilde{\kappa}:=\kappa(\tilde{\A}^\top \M) = O(1)$ holds with high probability. We first show that we can transform the problem of computing $\w_t^* = \tilde{\A}^\dagger \tilde{\b}_t$ to solving the following preconditioned normal equation of the second kind:
\begin{align*}
(\M^\top \tilde{\A}\tilde{\A}^\top \M) \u = \M^\top\tilde{\b}_t, ~~~ \w = \tilde{\A}^\top \M\u.
\end{align*}
Then we use CG to solve this positive semidefinite linear system over $\u$. Let $\u^* = \M^{-1}(\tilde{\A}\tilde{\A}^\top)^{\dagger} \tilde{\b}_t$ and $\tilde{\w}^* = \tilde{\A}^\top\M\u^*$ be the optimal solution, we show that $\tilde{\w}^* = \w_t^*$. By using Lemma~\ref{lem:cg} and doing the change of variable back from $\u$ to $\w$, we show in Lemma~\ref{lem:approx_converge} that after $s$ iterations,
\begin{align*}
\|\w_s - \w_t^*\| \leq 2\left(\frac{\tilde{\kappa} - 1}{\tilde{\kappa} + 1}\right)^s \|\w^*_t\|.
\end{align*}
The only part left is to construct a preconditioner $\M \in \R^{\tau \times \tau}$ such that $\tilde{\kappa} = O(1)$, and we construct it by using fast oblivious subspace embedding matrix proposed by \cite{chenakkod2023optimal}.
According to Lemma~\ref{lem:precondition_sparse}, by setting $\phi = O(\tau+\log(1/\delta))$ we have that $\tilde{\kappa} = O(1)$ holds with probability $1- \delta$, since we can set $\epsilon$ to be a small constant. Next we analyze the time complexity: given $\tilde{\A}$, computing $\mathbf{\Theta}\tilde{\A}^\top$ takes $O(n\tau \log (\tau/\delta) + \tau^2\log^4(\tau/\delta))$, computing the SVD of $\mathbf{\Theta}\tilde{\A}^\top$ takes $O(\phi \tau^{\omega-1})$, and computing $\M$ takes $O(\tau^\omega)$. Thus we conclude that the total time complexity is
\begin{align*}
O(n\tau \log (\tau/\delta) + \tau^2\log^4(\tau/\delta)+ \tau^{\omega-1}\phi) = O(n\tau \log(\tau/\delta) + \tau^{\omega} +\tau^{\omega-1}\log(1/\delta)) = O(n\tau \log(\tau/\delta) + \tau^{\omega}).
\end{align*}
With all the above analysis, we have the following convergence lemma for approximate sketch-and-project. For the proof of Lemma~\ref{lem:main_convergence} see Section~\ref{sec:proof_approx_project}.

\begin{lemma}[Convergence of approximate sketch-and-project]\label{lem:main_convergence}
Given matrix $\A\in\R^{m\times n}$ with rank $r$, denote $\mathbf{Q} = \frac{1}{\sqrt{m}}\H \mathbf{D}$ as the RHT matrix and $\bar{\A} = \Q\A$ as the preprocessed matrix. Let $\I_{\tilde{S}}\in\R^{\tau \times m}$ be the sampling matrix and $\P_{\tilde{S}} := \bar{\A}^\top \I_{\tilde{S}}^\top(\I_{\tilde{S}}\bar{\A}\bar{\A}^\top\I_{\tilde{S}}^\top)^\dagger\I_{\tilde{S}}\bar{\A}$ be the corresponding projection matrix. Denote $\rho = \lambda_{\min}^+(\E[\P_{\tilde{S}}])$, then, running Algorithm~\ref{alg:main} with choice $s_{\max} = O(\log\frac{1}{\rho})$ will yield the following convergence result:
\begin{align*}
\E\|\x_{t+1} - \x^*\|^2 \leq \left(1 - \frac{\rho}{4}\right)\cdot \|\x_t - \x^*\|^2.
\end{align*}
\end{lemma}
Since we already obtain a lower bound for the quantity $\rho = \lambda_{\min}^+(\E[\P_{\tilde{S}}])$ in Lemma~\ref{lem:main_sampling}, by combining it with Lemma~\ref{lem:main_convergence} we are ready to prove our main result.

\subsection{Proof of Theorem~\ref{thm:main}}
\begin{proof}[Proof of Theorem~\ref{thm:main}]
Notice that in Algorithm~\ref{alg:main} we only carry out randomized Hadamard transform once. By combining Lemma~\ref{lem:main_sampling} and Lemma~\ref{lem:main_convergence} we know that conditioned on an event that holds with probability $0.99$, if we set $s_{\max} = O(\log \frac{1}{\rho}) = O(\log(r\bar{\kappa}_k/\tau))$, then we have
\begin{align*}
\E\|\x_t - \x^*\|^2 \leq \left(1 - \frac{\rho}{4}\right)^t \cdot \|\x_0 - \x^*\|^2 \leq \left(1 - \frac{\tau / r}{10\bar{\kappa}_k^2\log^3 m}\right)^t \cdot \|\x_0 - \x^*\|^2.
\end{align*}
By further choosing $\x_0 = \mathbf{0}, \tau = O(k\cdot \log^3 m)$ and $t_{\max} = O(r\bar{\kappa}_k^2/k \cdot \log(100/\epsilon))$ and denoting $\tilde{\x} = \x_{t_{\max}}$, we have $\E\|\tilde{\x} - \x^*\|^2 \leq \frac{\epsilon}{100} \|\x_0 - \x^*\|^2 = \frac{\epsilon}{100} \|\x^*\|^2$. By using Markov's inequality and applying a union bound, we have that $\|\tilde{\x} - \x^*\|^2 \leq \epsilon \| \x^*\|^2$ holds with probability $0.98$, and the total time complexity is
\begin{align*}
& ~ O(mn\log m) + t_{\max} \cdot O(n\tau \log(\tau/\delta_1) + \tau^{\omega} + n\tau s_{\max}) \\
= & ~ O(mn\log m) + r\bar{\kappa}_k^2/k \cdot \log(100/\epsilon) \cdot O(nk \log^3 m \log (r\bar{\kappa}_k) + k^\omega \log^{3\omega} m) \\
= & ~ O(mn \log m + (nr\log^3 m \log (r\bar{\kappa}_k) + rk^{\omega-1} \log^{3\omega}m )\cdot\bar{\kappa}_k^2\log1/\epsilon) \\
= & ~ O(mn \log m + (nr\log^4 m + rk^{\omega-1} \log^{3\omega}m)\cdot\bar{\kappa}_k^2\log1/\epsilon)
\end{align*}
where we use the fact that $\delta_1 = \rho / 324 = \Omega(k / r)$ as stated in the proof of Lemma~\ref{lem:main_convergence} in Section~\ref{sec:proof_approx_project}, and the assumption $\bar{\kappa}_k \leq m^4$ which gives that $\log (r\bar{\kappa}_k) = O(\log m)$.
\end{proof}

\section{Expected Projection under DPP Sampling}
\label{sec:convergence}
In this section we give a lower bound for the quantity $\lambda_{\min}^+(\E[\P_{S_\dpp}])$ where $S_{\dpp}\sim k'$-$\dpp(\A\A^\top)$. This quantity is directly associated with the convergence rate of approximated sketch-and-project (see Lemma~\ref{lem:main_convergence}), and later we will show that $\rho = \lambda_{\min}^+(\E[\P_{\tilde{S}}]) \geq \lambda_{\min}^+(\E[\P_{S_\dpp}])$, where $\tilde S$ is the sampling scheme used in our main algorithms. 

\subsection{Elementary Symmetric Polynomials}
We start by presenting the following two lemmas which help us relate the expectation formulas that arise in sampling from $\dpp(\L)$ (for some psd matrix $\L$) to properties of elementary symmetric polynomials defined by the eigenvalues of $\L$. These lemmas will then be used in the proof of Lemma~\ref{lem:improve_bound}, the main result of this section. 

First we introduce Lemma~\ref{lem:sum_adj} which is later used to compute the numerator of Eq.\eqref{eq:lemma_combine}, in terms of elementary symmetric polynomials of eigenvalues. Recall that we use $\adj(\L)\in\R^{n\times n}$ to denote the adjugate of an $n\times n$ matrix $\L$, and note that when $\L$ is invertible, then we have $\adj(\L)=\det(\L)\L^{-1}$, which is how the lemma connects to computing expectations for DPP sampling.
\begin{lemma}[Lemma 3.3 of \cite{rk20}]
\label{lem:sum_adj}
Let $\L\in \R^{m\times m}$ be a real symmetric matrix with eigenvalues $\lambda = (\lambda_1, \ldots, \lambda_m)$ where $\lambda_1 \geq \cdots \geq \lambda_m \geq 0$, let $\L = \Q\diag(\lambda)\Q^\top$ be its eigendecomposition where $\Q\in\R^{m\times m}$ is orthogonal, and let $1\leq \ell \leq m$. Then we have
\begin{align*}
\sum_{S \in \tbinom{[m]}{\ell}} \I_S \cdot\adj(\L_{S, S})\cdot\I_S^\top = \Q\diag(\sigma_{\ell-1}(\lambda_{-1}), \ldots, \sigma_{\ell-1}(\lambda_{-m}))\Q^\top
\end{align*}
where $\lambda_{-i}$ denotes the vector $\lambda$ without the $i$-th element for $1\leq i \leq n$.
\end{lemma}
\noindent
Next we introduce Lemma~\ref{lem:poly_upper_bound} which provides an upper bound for the ratio of elementary symmetric polynomials by using Schur-concavity property of this quantity. Lemma~\ref{lem:poly_upper_bound} will be used in Eq.\eqref{eq:sym_poly_ratio}.
\begin{lemma}[Lemma 3.1 of \cite{gs12}]
\label{lem:poly_upper_bound}
Given vector $\lambda = (\lambda_1, \ldots, \lambda_n)$ with $\lambda_1 \geq \lambda_2 \geq \cdots \geq \lambda_n \geq 0$, given positive integers $\tau \geq j > 0$, we have
\begin{align*}
\frac{\sigma_{\tau + 1}(\lambda)}{\sigma_{\tau}(\lambda)} \leq \frac{1}{\tau+1-j} \sum_{i = j+1}^n \lambda_i.
\end{align*}
\end{lemma}
\noindent

\subsection{Proof of Lemma~\ref{lem:improve_bound}}
With the above two lemmas, we are ready to prove Lemma~\ref{lem:improve_bound}, which gives a lower bound for $\lambda_{\min}^+(\E[\P_{S_{\dpp}}])$ when $S_{\dpp}\sim k'$-$\dpp(\A\A^\top)$.

\begin{proof}[Proof of Lemma~\ref{lem:improve_bound}]
For simplicity we denote $S=S_{\dpp}\sim k'$-$\dpp(\A\A^\top)$, then the quantity we consider is $\lambda_{\min}^+(\A^\top \E[\I_S^\top (\I_S\A\A^\top\I_S^\top)^\dagger\I_S]\A)$. Denote $\B = \A\A^\top$, notice that $\B \in\R^{m \times m}$ is a PSD matrix with $\rank(\B) = r$. Denote $\{\lambda_i\}_{i=1}^m$ as the eigenvalues of $\B$, then we have $\lambda_i = \sigma_i^2$ for $i \leq r$ and $\lambda_i = 0$ for $r+1 \leq i \leq m$. 
Denote the singular value decomposition of $\A$ as $\A = \U\mathbf{\Sigma}\V^\top$ where $\U \in \R^{m \times m}, \mathbf{\Sigma} \in \R^{m\times n}$ and $\V \in \R^{n\times n}$,  then we have eigendecomposition $\B = \U(\mathbf{\Sigma}\mathbf{\Sigma}^\top) \U^\top$. According to \cite{rk20}, assuming for simplicity that all the $k' \times k'$ submatrices of $\B$ are non-degenerate (the other case can be reduced to this), by using Cramer's rule $\det(\B_{S, S})\cdot (\B_{S, S})^{-1} = \adj(\B_{S, S})$ we have
\begin{align}\label{eq:lemma_combine}
\E[\P_{S_{\dpp}}] = \A^\top\E[\I_S^\top (\I_S\B\I_S^\top)^\dagger\I_S]\A = \frac{\sum_{S \in \tbinom{[m]}{k'}} \A_{S,:}^\top \cdot\adj\left(\B_{S, S}\right) \cdot\A_{S,:}}{\sum_{S \in \tbinom{[m]}{k'}} \det\left(\B_{S, S}\right)}.
\end{align}
If we do not assume that all the $k' \times k'$ submatrices are non-degenerate, then the numerator becomes $\sum_{S \in \tbinom{[m]}{k'}: \det(\B_{S,S})\neq 0} \I_S \cdot\adj\left(\B_{S, S}\right) \cdot\I_S^T$. Then, for any vector $\w\in\R^n$ we have
\begin{align}\label{eq:degenerate}
\w^\top\E[\P_{S_{\dpp}}]\w = & ~ \frac{\sum_{S \in \tbinom{[m]}{k'}} \det(\B_{S,S})\cdot \w^\top\A_{S,:}^\top(\B_{S,S})^\dagger\A_{S,:}\w}{\sum_{S \in \tbinom{[m]}{k'}} \det\left(\B_{S, S}\right)} \nonumber \\
= & ~ \frac{\sum_{S \in \tbinom{[m]}{k'}: \det(\B_{S,S})\neq 0} \w^\top\A_{S,:}^\top\adj(\B_{S,S})\A_{S,:}\w}{\sum_{S \in \tbinom{[m]}{k'}} \det\left(\B_{S, S}\right)}
\end{align}
where the second step also follows from Cramer's rule. Notice that for any vector $\u$ and matrix $\C$ we have the following formula for adjugate matrix (see Fact 2.14.2 in \cite{bernstein2009matrix}):
\begin{align*}
\det(\C+\u\u^\top) = \det(\C) + \u^\top\adj(\C)\u.
\end{align*}
By applying this formula, for any $S$ such that $\det(\B_{S,S}) = 0$, we can compute each component of numerator in Eq.\eqref{eq:degenerate} as
\begin{align*}
\w^\top\A_{S,:}^\top\adj(\B_{S,S})\A_{S,:}\w = & ~ \det(\B_{S,S} + \A_{S,:}\w\w^\top\A_{S,:}^\top) - \det(\B_{S,S}) \\
= & ~ \det(\A_{S,:}\A_{S,:}^\top + \A_{S,:}\w\w^\top\A_{S,:}^\top) \\
= & ~ \det(\A_{S,:}(\I+\w\w^\top)\A_{S,:}^\top) = 0
\end{align*}
where the last step follows since $\rank(\A_{S,:}(\I+\w\w^\top)\A_{S,:}^\top) \leq \rank(\A_{S,:}) < k'$, according to our assumption that $\det(\B_{S,S}) = \det(\A_{S,:}\A_{S,:}^\top) = 0$. By applying this result to Eq.\eqref{eq:degenerate} we have 
\begin{align*}
\w^\top\E[\P_{S_{\dpp}}]\w = & ~ \frac{\sum_{S \in \tbinom{[m]}{k'}: \det(\B_{S,S})\neq 0} \w^\top\A_{S,:}^\top\adj(\B_{S,S})\A_{S,:}\w}{\sum_{S \in \tbinom{[m]}{k'}} \det\left(\B_{S, S}\right)} \\
= & ~ \frac{\sum_{S \in \tbinom{[m]}{k'}} \w^\top\A_{S,:}^\top\adj(\B_{S,S})\A_{S,:}\w}{\sum_{S \in \tbinom{[m]}{k'}} \det\left(\B_{S, S}\right)}
\end{align*}
which shows that the degenerate case can indeed be reduced to the non-degenerate case as in Eq.\eqref{eq:lemma_combine}, since we only care about the smallest positive eigenvalue of $\E[\P_{S_{\dpp}}]$. According to Lemma~\ref{lem:sum_minor} and Lemma~\ref{lem:sum_adj}, we can compute the numerator and denominator separately and have
\begin{align}\label{eq:eigen_decompose}
\E[\I_S^\top (\I_S\B\I_S^\top)^\dagger\I_S] = \frac{\U \diag\left(\sigma_{k'-1}\left(\lambda_{-1}\right), \ldots, \sigma_{k'-1}\left(\lambda_{-m}\right)\right) \U^\top}{\sigma_{k'}(\lambda)}
\end{align}
where $\sigma_{k'}$ is the real elementary symmetric polynomial with degree $k'$, and $\lambda_{-i} \in \R^{m-1}$ denotes the vector $\lambda$ without the $i$-th element.
Since we want to lower bound the above quantity, we need to lower bound each $\frac{\sigma_{k' - 1}(\lambda_{-i})}{\sigma_{k'}(\lambda)}$ for $i\in[m]$. We apply Lemma~\ref{lem:poly_upper_bound} and know for all $k' > k$ we have
\begin{align*}
\frac{\sigma_{k'}(\lambda)}{\sigma_{k'-1}(\lambda)} \leq \frac{1}{k'-k} \sum_{j = k+1}^m \lambda_j.
\end{align*}
Thus for all $i \in [m]$ and $k' > k$, we have
\begin{align}\label{eq:sym_poly_ratio}
\frac{\sigma_{k' - 1}(\lambda_{-i})}{\sigma_{k'}(\lambda)} =  \frac{\sigma_{k' - 1}(\lambda_{-i})}{\sigma_{k'}(\lambda_{-i}) + \lambda_i \cdot \sigma_{k'-1}(\lambda_{-i})} = \frac{1}{\lambda_i + \frac{\sigma_{k'}(\lambda_{-i})}{\sigma_{k'-1}(\lambda_{-i})}} \geq \frac{1}{\lambda_i + \frac{1}{k' - k} \sum_{j > k, j \neq i} \lambda_j}.
\end{align}
We consider two different cases for the value of $i$ and $k$, and have
\begin{align*}
\begin{cases}
\frac{\sigma_{k' - 1}(\lambda_{-i})}{\sigma_{k'}(\lambda)} \geq \frac{1}{\lambda_i + \frac{1}{k' - k} \sum_{j > k} \lambda_j} ~~~\text{for}~~~ i \leq k; \\
\frac{\sigma_{k' - 1}(\lambda_{-i})}{\sigma_{k'}(\lambda)} \geq \frac{1}{\frac{k'-k-1}{k'-k}\lambda_i + \frac{1}{k' - k} \sum_{j > k} \lambda_j} ~~~\text{for}~~~ k+1 \leq i \leq m.
\end{cases}
\end{align*}
Based on these two cases, we construct the following approximation matrix:
\begin{align*}
\B_{k'} := \U \diag\left(\lambda_1, \ldots, \lambda_{k}, \frac{k'-k-1}{k'-k}\lambda_{k+1}, \ldots, \frac{k'-k-1}{k'-k}\lambda_m\right) \U^\top + \frac{1}{k'-k}\sum_{j>k}^m \lambda_j\I
\end{align*}
an have $\E[\I_S^\top (\I_S\B\I_S^\top)^\dagger\I_S] \succeq \B_{k'}^{-1}$. By applying this result we can lower bound $\lambda_{\min}^+(\E[\P_{S}])$.
\begin{align*}
\lambda_{\min}^+(\E[\P_{S}]) = & ~\lambda_{\min}^+(\A^\top\E[\I_S^\top (\I_S\B\I_S^\top)^\dagger\I_S] \A) \\
\geq & ~ \lambda_{\min}^+(\A^\top\B_{k'}^{-1} \A) \\
= & ~ \lambda_{\min}^+(\V\mathbf{\Sigma}^\top\U^\top \B_{k'}^{-1} \U\mathbf{\Sigma}\V^\top) \\
= & ~ \lambda_{\min}^+(\V\diag(\frac{\lambda_1}{\lambda_1+\frac{1}{k'-k}\sum_{j>k}^m\lambda_j}, \ldots, \frac{\lambda_{k}}{\lambda_{k}+\frac{1}{k'-k}\sum_{j>k}^m\lambda_j}, \\
& ~~~~~~~~\frac{\lambda_{k+1}}{\frac{k'-k-1}{k'-k}\lambda_{k+1}+\frac{1}{k'-k}\sum_{j>k}^m\lambda_j}, \ldots, \frac{\lambda_n}{\frac{k'-k-1}{k'-k}\lambda_n+\frac{1}{k'-k}\sum_{j>k}^m\lambda_j})\V^\top) \\
= & ~ \frac{\sigma_{r}^2}{\frac{k'-k-1}{k'-k}\sigma_r^2+\frac{1}{k'-k}\sum_{j>k}^r\sigma_j^2}
\end{align*}
where the last step follows from $\lambda_j = \sigma_j^2$ for $j \in [r]$, and that $\lambda_{r+1} = \cdots = \lambda_m = 0$. Recall that we define $\bar{\kappa}_k := (\frac{1}{r-k}\sum_{j > k}^r \sigma_j^2 / \sigma_r^2)^{1/2}$, thus we can simplify the above expression as
\begin{align*}
\lambda_{\min}^+(\E[\P_{S}]) \geq \frac{1}{\frac{k'-k-1}{k'-k}+\frac{1}{k'-k}\sum_{j>k}^r\sigma_j^2 / \sigma_r^2} = \frac{k'-k}{k'-k-1 + (r-k)\bar{\kappa}_{k}^2}.
\end{align*}
\end{proof}

\paragraph{Discussion.}
Let $\B = \A\A^\top \in \R^{m\times m}$ with rank $r$, in comparison with our result, \cite{rk20} constructed the ``$k'$-coordinate approximation'' of $\B$ as
\begin{align*}
\hat{\B}_{k'} := \U \diag(\lambda_1, \ldots, \lambda_{k'}, \lambda_{k'}, \ldots, \lambda_{k'})\U^\top + \sum_{j=k'+1}^m \lambda_j \I
\end{align*}
and showed that $\mathbb{E} [\I_S^\top(\I_S\B\I_S^\top)^\dagger \I_S] \succeq \hat{\B}_{k'}^{-1}$. By using their result, one could show the following weaker lower bound:
\begin{align}\label{eq:proj_lower_bound_rk20}
\lambda_{\min}^+(\E[\P_{S}]) \geq \lambda_{\min}^+(\A^\top\hat{\B}_{k'}^{-1} \A) = \frac{\lambda_r}{\sum_{j=k'}^m\lambda_j} = \frac{\sigma_r^2}{\sum_{j=k'}^r \sigma_j^2} = \frac{1}{(r-k'+1)\bar{\kappa}_{k'-1}^2}.
\end{align}
By comparing this result with the lower bound we get in Eq.\eqref{eq:proj_lower_bound}, we can see that we have an extra parameter $k$ that only needs to satisfy $k< k'$. If we choose $k = k'-1$ then we recover the same result as in \cite{rk20}. However, if we consider the setting where $k$ is such that $\frac{\sigma_i}{\sigma_r} = O(1)$ holds for all $i \geq k$ and assume that $k \leq \frac{r}{2}$, then by choosing $k' = 2k-1 <r$, Eq.\eqref{eq:proj_lower_bound_rk20} only gives a lower bound of order $\Omega(\frac{1}{r-k})$. This is worse by a factor of $k$ than our result which is of order $\Omega(\frac{k}{r})$.

\section{Replacing DPP with Uniform Sampling}\label{sec:sampling}
In Lemma~\ref{lem:improve_bound} we give a lower bound for $\lambda_{\min}^+(\E[\P_{S}])$ when $S$ is a $k$-DPP sample.\footnote{Note that in this section we use $k$ to denote DPP size for simplicity, which corresponds to $k'$ in previous sections.} However, obtaining an exact DPP sample can be expensive. Currently, the best known result of obtaining $s$ independent approximate samples takes $\tilde{O}(mk^{\omega-1} + sk^{\omega})$, see Lemma~\ref{lem:anari_main_1}, and adapting this result to our rectangular matrix setting takes $\tilde{O}(mnk^{\omega-2} + snk^{\omega})$, see Section~\ref{sec:compare_sample_dpp}. This sampling method consists of two steps: a preprocessing step which takes $\tilde{O}(mnk^{\omega-2})$ and a sampling step which takes $\tilde{O}(nk^{\omega})$ for each sample. Notice that \emph{both these two steps are bottlenecks and unaffordable} when we use the sketch-and-project method with $k$-DPP sampling, since with the outer iteration $t_{\max} = O(r\bar{\kappa}_k^2/k \cdot\log(1/\epsilon))$, the cost will be $\tilde{O}(mnk^{\omega-2} + nr\bar{\kappa}_k^2k^{\omega-2}\log(1/\epsilon))$, i.e., at least as much as standard preconditioning for CG via power iteration.

To conquer the first bottleneck, we use a different method of preprocessing. The purpose of the preprocessing step is to make the marginal probabilities $\Pr\{j\in S_{\dpp}\}$ almost the same. \cite{alv22} do this by estimating all of the marginal probabilities, which itself requires repeatedly sampling from the $k$-DPP distribution. We instead show in Lemma~\ref{lem:approx_iso_fixed} that the randomized Hadamard transform (RHT) can uniformize the marginals \emph{in near-linear time} $O(mn\log m)$, which is much better than $\tilde{O}(mnk^{\omega-2})$. With the approximately uniform marginals, we can use the following result from \cite{alv22} to obtain an approximated $k$-DPP sample.
\begin{lemma}[Theorem 3 in \cite{alv22}]
\label{lem:anari_main_2}
Given a strongly Rayleigh distribution $\mu \in \R^{\tbinom{[m]}{k}}$ and marginal overestimates $q_i \geq \Pr_{T\sim\mu}\{i\in T\}$ for $i\in[m]$. Let $K := \sum_{i\in [m]} q_i$, there is an algorithm that outputs a sample from a distribution with total variation distance $m^{-O(1)}$ from $\mu$ in time bounded by $O(\log^3 m)$ calls to $\mathcal{T}_\mu(O(K), k)$, where $\mathcal{T}_\mu(t, k)$ is the time it takes to produce a sample from $\mu$ conditioned on all elements of the sample being a subset of a fixed set $T$ of size $|T| = t$. 
\end{lemma}
\begin{remark}\label{rem:anari_constant}
In Algorithm~\ref{alg:dpp_sample}, we give the sampling procedure claimed by Lemma~\ref{lem:anari_main_2} under the additional ``isotropic'' assumption, i.e., that the marginal overestimates are all uniform. Moreover, the $m^{-O(1)}$ term in Lemma~\ref{lem:anari_main_2} is specified as $m^{-9}$ in the proof of \cite{alv22}.
\end{remark}

\begin{algorithm}[!ht]
\caption{$k$-DPP sampling algorithm based on \cite{alv22} (Ghost Algorithm).}
\label{alg:dpp_sample_ghost}
\begin{algorithmic}[1]
\State \textbf{Input: }$m, k, t, \mathrm{iter}, \mathrm{call}, \mathcal{O}$;
\State \textbf{Initialize }$S_0 \in  \tbinom{[m]}{k}, T_0 \gets \varnothing$; \Comment{By calling $\tilde{O}(k)$ times the oracle.}
\For{$i = 0, 1, 2, \ldots, \mathrm{iter}-1$} \Comment{$\mathrm{iter} = O(\log^3 m)$.}
\State \textbf{Initialize }$T_{i, 0} \gets T_i$;
\For{$j = 0, 1, 2, \ldots, \mathrm{call}-1$} \Comment{$\mathrm{call} \geq Ck$ for some $C=O(1)$.}
\State Sample $i_j$ by calling $ \mathcal{O}$;
\If{$i_j \in S_i \cup T_{i,j}$}:
\State $T_{i, j+1} = T_{i, j}$;
\Else: 
\State $T_{i, j+1} = T_{i,j} \cup \{i_j\}$;
\EndIf
\EndFor
\State \textbf{Update }$T_{i+1} \gets T_{i, \mathrm{call}}$; \Comment{Up-walk step by calling $ \tilde{O}(t-k)$ times the oracle.}
\State \textbf{Update }$S_{i+1} \subseteq S_i \cup T_{i+1}$; \Comment{Down-walk sampling from a $k$-DPP restricted to $S_i \cup T_{i+1}$.}
\EndFor \\
\Return $S_{\mathrm{iter}}$.
\end{algorithmic}
\end{algorithm}

To conquer the second bottleneck, we observe that compared with a real $k$-DPP sample, we are actually more satisfied with a larger set that \emph{contains} a $k$-DPP in the sketch-and-project scheme, since it can guarantee an even faster convergence rate, see Lemma~\ref{lem:project_order}.
How do we obtain a sample that contains a $k$-DPP as a subset? We show that we can achieve this by using \emph{only uniform sampling}. Suppose we have an oracle $\mathcal{O}$ that can uniformly sample one element from $[m]$, and we would like to obtain $\tilde{S}$ such that $\tilde{S}$ contains a $k$-DPP sample. Our idea is to rewrite Algorithm \ref{alg:dpp_sample} so that it can only access the elements via the oracle $\mathcal{O}$: starting from $S_0$ with size $k$, we call the oracle several times, if the sample is new then we add it to $S_0$, otherwise we discard it. After we obtain $S_0$, we follow the same rule to obtain $T_1$ such that (i) $|T_1| = t - k$ and (ii) $T_1 \cap S_0 = \varnothing$. Then, we down sample $S_1$ such that (i) $|S_1| = k$ and (ii) $S_1 \subseteq S_0 \cup T_1$. These two steps finish one iteration. Notice that only step 1 needs to call the oracle, while step 2 does not.

These steps correspond to Algorithm~\ref{alg:dpp_sample_ghost}, and we call it the ``Ghost Algorithm'' since we do not implement it due to the high cost of the down-walk step (line $14$ of Algorithm~\ref{alg:dpp_sample_ghost}). Instead, we use the coupled Algorithm~\ref{alg:dpp_sample_real} and show that the output of Algorithm~\ref{alg:dpp_sample_real} contains the output of Algorithm~\ref{alg:dpp_sample_ghost} as a subset. By further showing that the output of Algorithm~\ref{alg:dpp_sample_ghost} is the same as a $k$-DPP sample with high probability, we have the following theorem which shows that the output of Algorithm~\ref{alg:dpp_sample_real} indeed contains a $k$-DPP sample. Notice that since we only use uniform sampling in our algorithm, the total cost of sampling takes only $O((k+\log m) \log^3 m) = \tilde{O}(k)$ which is much better than $\tilde{O}(nk^\omega)$ and almost linear in $k$. For detailed discussion see Section~\ref{sec:compare_sample_dpp}. In this whole section,  we denote $\mathbf{Q} = \frac{1}{\sqrt{m}} \H \mathbf{D}$ as the RHT matrix, where $\H \in\R^{m \times m}$ is the Hadamard matrix and $\mathbf{D} \in \R^{m\times m}$ is the diagonal matrix with Rademacher entries. We denote $\bar{\A} = \Q\A$ as the preprocessed matrix.

\begin{lemma}[Coupling]\label{lem:main_coupling}
Given matrix $\A\in\R^{m\times n}$, denote $\mathbf{Q} = \frac{1}{\sqrt{m}}\H \mathbf{D}$ as the RHT matrix and $\bar{\A} = \Q\A$ as the preprocessed matrix. Suppose we have an oracle $\mathcal{O}$ that can uniformly sample one element from $[m]$. Conditioned on an event that holds with probability $0.99$ (and only depends on RHT), for $k \geq \log m$ and $\delta = m^{-O(1)}$, there is a coupling $(\tilde{S}, S_{\dpp})$ where $\tilde{S}$ is sampled according to Algorithm~\ref{alg:dpp_sample_real} with choice $\tau \geq Ck\log^3 m$ for some constant $C=O(1)$, and $S_{\dpp}\sim k$-$\dpp(\bar{\A}\bar{\A}^\top)$, such that with probability $1 - \delta$, we have $S_{\dpp} \subseteq \tilde{S}$.
\end{lemma}

\begin{algorithm}[!ht]
\caption{Sample by using uniform sampling oracle $\mathcal{O}$.}
\label{alg:dpp_sample_real}
\begin{algorithmic}[1]
\State \textbf{Input: }$m, \tau, \mathcal{O}$;
\State \textbf{Initialize }$\tilde{S}_0 \gets \varnothing$;
\For{$\ell = 0, 1, 2, \ldots, \tau-1$}
\State Sample $i_\ell$ by calling $\mathcal{O}$;
\If{$i_\ell \in \tilde{S}_\ell$}:
\State $\tilde{S}_{\ell+1} = \tilde{S}_\ell$;
\Else: 
\State $\tilde{S}_{\ell+1} = \tilde{S}_\ell \cup \{i_\ell\}$;
\EndIf
\EndFor \\
\Return $\tilde{S} = \tilde{S}_N$.
\end{algorithmic}
\end{algorithm}

To prove Lemma~\ref{lem:main_coupling}, we show the following in this section:
\begin{enumerate}
    \item After the preprocessing step, the marginals are approximately isotropic. See Lemma~\ref{lem:approx_iso_fixed}.
    \item In Algorithm~\ref{alg:dpp_sample_ghost}, with probability $1-\delta$, we can obtain each $T_i \in \tbinom{[m]}{t-k}$ after $O(t-k+\log\frac{1}{\delta})$ calls to the oracle. See Lemma~\ref{lem:sample_log}.
    \item In Algorithm~\ref{alg:dpp_sample_ghost}, if we set $\mathrm{call} \geq Ck$ for some constant $C=O(1)$ and $\mathrm{iter} = O(\log^3 m)$, then with probability $1 - \delta$, $S_{\mathrm{iter}}$ is the same as $S_{\dpp}$. See Lemma~\ref{lem:ghost_coupling}.
\end{enumerate}

\paragraph{Monotonicity of Projections.}
We start our analysis by making the following simple observation, which implies that in the sketch-and-project scheme, a larger set can only improve the convergence guarantee, since it leads to a ``larger'' projection matrix in the sense of matrix orders. Formally we have the following lemma.
\begin{lemma}
\label{lem:project_order}
Given matrix $\A\in\R^{m \times n}$, for any subsets $S_1 \subseteq S_2 \subseteq [m]$, define projection matrix $\P_{S_i} := \A^\top \I_{S_i}^\top (\I_{S_i}\A\A^\top\I_{S_i}^\top)^\dagger\I_{S_i}\A$ where $\I_{S_i} \in \R^{|S_i| \times m}$ is the sampling matrix. Then we have
\begin{align*}
\P_{S_1} \preceq \P_{S_2}.
\end{align*}
\end{lemma}

\begin{proof}
We only need to show that for any vector $\x \in \R^n$ such that $\|\x\|=1$, we have $\x^\top \P_{S_1} \x \leq \x^\top\P_{S_2} \x$. Let $\mathrm{span}\langle\A_{S_i}\rangle$ be the subspace spanned by the rows of $\I_{S_i}\A$. We have
\begin{align*}
\x^\top \P_{S_i} \x 
= 1 - \x^\top(\I-\P_{S_i})\x
= 1- \|\x - \P_{S_i}\x\|^2 = 1 - \min_{\x_i \in \mathrm{span}\langle\A_{S_i}\rangle} \|\x - \x_i\|^2 .
\end{align*}
 Note that since $S_1 \subseteq S_2$, we know that $\mathrm{span}\langle\A_{S_2}\rangle$ contains $\mathrm{span}\langle\A_{S_1}\rangle$, thus we have:
\begin{align*}
 \min_{\x_1 \in \mathrm{span}\langle\A_{S_1}\rangle} \|\x - \x_1\|^2 \geq \min_{\x_2 \in \mathrm{span}\langle\A_{S_2}\rangle} \|\x- \x_2\|^2.
\end{align*}
We conclude that $\x^\top \P_{S_1} \x \leq \x^\top \P_{S_2} \x$ holds for all $\x \in \R^n$ such that $\|\x\|=1$, as desired.
\end{proof}

According to Lemma~\ref{lem:project_order}, $\P_{S_1} \preceq \P_{S_2}$ holds for all $S_1 \subseteq S_2$. Thus we can take expectation over the randomness of $S_1$ and $S_2$ and have $\lambda_{\min}^+(\E[\P_{S_1}]) \leq \lambda_{\min}^+(\E[\P_{S_2}])$ holds conditioned on the event $S_1 \subseteq S_2$. 
Thus, if we can obtain a sample $\tilde{S}$ that contains a $k$-DPP sample $S_{\dpp}$, then we know that $\lambda_{\min}^+(\E[\P_{\tilde{S}}]) \geq \lambda_{\min}^+(\E[\P_{S_\dpp}])$ which guarantees a faster convergence rate for our algorithm.

\subsection{Preprocessing}
In this section we show that the $k$-DPP marginals become approximately isotropic (i.e., uniform) after we use the randomized Hadamard transform (RHT) as a preprocessing step. We note that prior works have shown that the RHT can uniformize the entries of a vector, or the row-norms or leverage scores of a matrix \cite{t11}. However, this analysis typically relies on the fact that all of these quantities can be easily described as linear or quadratic functions of the rows of the target matrix. On the other hand, the marginals of a general fixed size DPP are not expressed with any simple formula (they can only be described with expressions that involve elementary symmetric polynomials of certain sub-matrices of the input matrix). We get around this issue by relying on the fact that $k$-DPP marginals can be described as a mixture of the marginals of simpler Projection DPPs (this fact has been previously used in $k$-DPP sampling algorithms \cite{kt12}).

First, we show that applying the RHT as in Algorithm \ref{alg:main} does not affect the optimal solution of the linear system $\A\x = \b$. Denote $\bar{\A} = \Q\A, \bar{\b} = \Q\b$ and consider the linear system $\bar{\A}\x=\bar{\b}$. Since $\Q^\top\Q = \I$, we know the optimal solution is
\begin{align*}
\bar{\A}^\dagger\bar{\b} = (\Q\A)^\dagger\Q\b = (\A^\top\Q^\top\Q\A)^\dagger \A^\top\Q^\top\Q\b = (\A^\top\A)^\dagger\A^\top\b = \A^\dagger\b = \x^*,
\end{align*}
which is the same as the optimal solution of $\A\x = \b$. To construct this system, we need to compute $\Q\A$ and $\Q\b$ which takes $O(mn\log m) = \tilde{O}(mn)$ time, this is much better than the $\tilde{O}(mnk^{\omega-2})$ result in Section~\ref{sec:compare_sample_dpp}. Next we show that the marginals of $k$-$\dpp(\bar{\A}\bar{\A}^\top)$ are approximately isotropic.

\begin{lemma}[Approximately isotropic $k$-DPP]
\label{lem:approx_iso_fixed}
Given matrix $\A\in\R^{m\times n}$, denote $\mathbf{Q} = \frac{1}{\sqrt{m}}\H \mathbf{D}$ as the RHT matrix and $\bar{\A} = \Q\A$ as the preprocessed matrix. Let $S_{\dpp}\sim k$-$\dpp(\bar{\A}\bar{\A}^\top)$, for any $0 < \delta < 1$, with probability $1-\delta$ we have
\begin{align*}
\Pr\{j \in S_{\dpp}\} \leq \left(\sqrt{\frac{k}{m}} + \sqrt{\frac{8\log(m / \delta)}{m}}\right)^2 ~~~\text{for all}~~~ j \in [m].
\end{align*}
\end{lemma}

\begin{proof}
Although there is no explicit expression for the marginals of general fixed size DPP, there is one specific class of fixed size DPP called Projection DPP which has simple expressions for the marginals. Moreover, sampling a $k$-$\dpp$ is equivalent to sampling a mixture of Projection DPPs, see Lemma~\ref{lem:sample_dpp_projection}. Let $\A = \U\mathbf{\Sigma}\V^\top$ be the SVD of $\A$ and denote $\bar{\U} = \Q\U$, then we have eigendecomposition $\bar{\A}\bar{\A}^\top = \bar{\U}(\mathbf{\Sigma}\mathbf{\Sigma}^\top)\bar{\U}^\top$. Denote $r = \rank(\bar{\A}\bar{\A}^\top) = \rank(\A)$. Lemma~\ref{lem:sample_dpp_projection} shows that if we first sample $\gamma = (\gamma_1, \ldots, \gamma_r)$ where $\gamma_i\sim\mathrm{Bernoulli}(\frac{\lambda_i}{\lambda_i +1})$ are independent, and then sample $S_\gamma \sim \mathrm{P}$-$\dpp(\bar{\U}_{:,\gamma}\bar{\U}_{:,\gamma}^\top)$, where $\bar{\U}_{:,\gamma}$ is the submatrix of columns of $\bar{\U}$ indexed by the non-zeros in $\gamma$, then conditioned on the event that $\sum_{i} \gamma_i = k$, $S_\gamma$ is distributed the same as $S_{\dpp}\sim k$-$\dpp(\bar{\A}\bar{\A}^\top)$.

The advantage of decomposing $k$-DPP into $\mathrm{P}$-DPPs is that we can explicitly express marginals of $\mathrm{P}$-DPP. Since $\bar{\U}_{:,\gamma}\bar{\U}_{:,\gamma}^\top = \bar{\U}_{:,\gamma}\bar{\U}_{:,\gamma}^\dagger$ is a projection matrix,  we know that the marginals of $\mathrm{P}$-$\dpp(\bar{\U}_{:,\gamma}\bar{\U}_{:,\gamma}^\top)$ can be expressed as the diagonal entries of this kernel matrix as (e.g., see \cite{dm21})
\begin{align*}
\Pr\{j \in S_\gamma\} = (\bar{\U}_{:,\gamma}\bar{\U}_{:,\gamma}^\top)_{j,j} = \|\e_j^\top \bar{\U}_{:,\gamma}\|^2 ~~~\text{where}~~~ j\in[m].
\end{align*}
Based on this result, we can compute the marginals of $S_{\dpp}$ by summing over the marginals of all possible P-DPPs. For $j\in[m]$ we have
\begin{align*}
\Pr\{j \in S_{\dpp}\} = \sum_{\gamma \in \{0,1\}^r} \left(\Pr\{j \in S_\gamma\}\cdot\Pr\{\gamma \mid \sum_i \gamma_i = k\} \right) = \sum_{\gamma \in \{0,1\}^r} \left(\|\e_j^\top \bar{\U}_{:,\gamma}\|^2\cdot\Pr\{\gamma \mid \sum_i \gamma_i = k\} \right).
\end{align*}
Note that by definition we have
\begin{align*}
\|\e_j^\top \bar{\U}_{:,\gamma}\|^2 = \|\e_j^\top \Q\U_{:,\gamma}\|^2 = \frac{1}{m}\|\e_j^\top\H\mathbf{D}\U_{:,\gamma}\|^2 = \frac{1}{m}\|\mathbf{d}^\top \diag(\e_j^\top\H)\U_{:,\gamma}\|^2,
\end{align*}
where $\mathbf{d}$ is the vector of Rademacher variables along the diagonal of $\D$. If we denote $\mathbf{E}_j := \diag(\e_j^\top\H)$ and $p_{\gamma} := \Pr\{\gamma \mid \sum_i \gamma_i = k\}$, then we have
\begin{align*}
\Pr\{j \in S_{\dpp}\} = & ~ \sum_{\gamma \in \{0,1\}^r} p_{\gamma}\cdot \|\e_j^\top \bar{\U}_{:,\gamma}\|^2 \\
= & ~ \frac{1}{m} \sum_{\gamma \in \{0,1\}^r} p_{\gamma}\cdot \|\mathbf{d}^\top \mathbf{E}_j\U_{:,\gamma}\|^2. \\ 
= & ~ \frac{1}{m} \sum_{\gamma \in \{0,1\}^r} p_{\gamma}\cdot \mathbf{d}^\top\mathbf{E}_j\U_{:,\gamma}\U_{:,\gamma}^\top \mathbf{E}_j \mathbf{d} \\ 
= & ~ \frac{1}{m} \mathbf{d}^\top \left(\sum_{\gamma \in \{0,1\}^r} p_{\gamma} \mathbf{E}_j \U_{:,\gamma}\U_{:,\gamma}^\top\mathbf{E}_j\right)\mathbf{d}.
\end{align*}
We define function $f_j(\x) = \sqrt{\frac{1}{m} \x^\top\B_j \x}$ where $\B_j := \sum_{\gamma \in \{0,1\}^r} p_{\gamma} \mathbf{E}_j \U_{:,\gamma}\U_{:,\gamma}^\top\mathbf{E}_j$. We first analyze the properties of $\B_j$. Since $\mathbf{E}_j^2 = \I$ and $\U_{:,\gamma}^\top\U_{:,\gamma} = \I$, we have
\begin{align*}
\|\B_j\| \leq \sum_{\gamma \in \{0,1\}^r} p_{\gamma} \cdot\|\mathbf{E}_j \U_{:,\gamma}\U_{:,\gamma}^\top\mathbf{E}_j\| \leq \sum_{\gamma \in \{0,1\}^r} p_{\gamma} \cdot\| \U_{:,\gamma}\U_{:,\gamma}^\top\| = 1
\end{align*}
and
\begin{align*}
\tr(\B_j) = \sum_{\gamma \in \{0,1\}^r} p_{\gamma}\cdot \tr(\mathbf{E}_j \U_{:,\gamma}\U_{:,\gamma}^\top\mathbf{E}_j) = \sum_{\gamma \in \{0,1\}^r} p_{\gamma} \cdot \tr(\U_{:,\gamma}^\top\U_{:,\gamma}) = k\cdot\left(\sum_{\gamma \in \{0,1\}^r} p_{\gamma}\right) = k.
\end{align*}
Next we analyze the properties of $f_j$. Note that $f_j$ is convex, with Lipschitz constant as follows:
\begin{align*}
|f_j(\x) - f_j(\y)| = \frac{1}{\sqrt{m}}|\|\x\|_{\B_j} - \|\y\|_{\B_j}| \leq \frac{1}{\sqrt{m}} \|\x-\y\|_{\B_j} \leq \frac{1}{\sqrt{m}} \|\B_j\|\cdot \|\x-\y\| \leq \frac{1}{\sqrt{m}} \|\x-\y\|.
\end{align*}
We consider the random variable $f_j(\mathbf{d})$ and bound the expectation as follows:
\begin{align*}
\E[f_j(\mathbf{d})] \leq (\E[f_j(\mathbf{d})^2])^{1/2} = \left(\frac{1}{m} \E[\mathbf{d}^\top\B_j \mathbf{d}]\right)^{1/2} = \left(\frac{1}{m} \tr(\B_j)\right)^{1/2} = \sqrt{\frac{k}{m}}.
\end{align*}
Then we use Lemma~\ref{lem:rademacher} on the concentration of a Lipschitz function of a Rademacher random vector, with $t = \sqrt{8\log(m / \delta)}$, and obtain
\begin{align*}
\Pr\left\{\sqrt{\frac{1}{m} \mathbf{d}^\top\B_j \mathbf{d}} \geq \sqrt{\frac{k}{m}} + \sqrt{\frac{8\log(m/ \delta)}{m}} \right\} \leq \frac{\delta}{m}.
\end{align*}
Further notice that $\Pr\{j \in S_{\dpp}\} = f_j^2(\mathbf{d})= \frac{1}{m} \mathbf{d}^\top\B_j \mathbf{d}$, thus by applying union bound we have that with probability $1-\delta$,
\begin{align*}
\Pr\{j \in S_{\dpp}\} \leq \left(\sqrt{\frac{k}{m}} + \sqrt{\frac{8\log(m / \delta)}{m}}\right)^2 ~~~\text{for all}~~~ j \in [m].
\end{align*}
\end{proof}
\begin{remark}\label{rem:approx_iso}
Since we assume that $k \geq \log m$, by setting $\delta = 0.01$ we have
\begin{align*}
\Pr\{j \in S_{\dpp}\} \leq  \left(\sqrt{\frac{k}{m}} + \sqrt{\frac{8\log(m / \delta)}{m}}\right)^2 \leq \frac{2k}{m}+ \frac{16\log(m / \delta)}{m} \leq \frac{c_0 k}{m}
\end{align*}
for some constant $c_0=O(1)$. This shows that after the randomized Hadamard transform, the marginals are approximately isotropic up to a constant factor with probability $0.99$.
\end{remark}

\subsection{Successful DPP Sample}
In this section we show that as long as the marginals are approximately isotropic, then with high probability, the output of Algorithm~\ref{alg:dpp_sample_ghost} is the same as a real $k$-$\dpp(\bar{\A}\bar{\A}^\top)$ sample. The main result is Lemma~\ref{lem:ghost_coupling}. Notice that Algorithm~\ref{alg:dpp_sample_ghost} generates new $T_i$'s with given size by calling the uniform sampling oracle, thus we first compute the number of calls to the oracle we need in Lemma~\ref{lem:sample_log}. Since the proof of Lemma~\ref{lem:sample_log} is a straightforward concentration argument, we defer the proof to Appendix~\ref{sec:appendix_proof}.

\begin{lemma}\label{lem:sample_log}
Suppose we have an oracle that can uniformly sample one element from $[m]$. Suppose we have $S \in \tbinom{[m]}{k}$ where $k \leq m$. Let $t = O(k)$ and $0 < \delta < 1$, then with probability $1 - \delta$, after at most $O(t+\log\frac{1}{\delta})$ calls to the oracle, we can obtain $T := \cup_{j=1}^t \{i_j\}$ such that $T \subseteq [m] \backslash S$.
\end{lemma}

\begin{lemma}[Ghost coupling]\label{lem:ghost_coupling}
Conditioned on an event that holds with probability $0.99$ (and only depends on RHT). For $k \geq \log m$ and $\delta = m^{-O(1)}$, if we run Algorithm~\ref{alg:dpp_sample_ghost} with choice $\mathrm{call} \geq Ck$ for some constant $C = O(1)$, and $\mathrm{iter} = O(\log^3 m)$, then with probability $1 - \delta$, $S_{\mathrm{iter}}$ is the same as $S_{\dpp}$, where $S_{\dpp}\sim k$-$\dpp(\bar{\A}\bar{\A}^\top)$.
\end{lemma}
\begin{proof}
According to Lemma~\ref{lem:approx_iso_fixed} and Remark~\ref{rem:approx_iso}, with probability $0.99$ we have $\Pr\{j \in S_{\dpp}\} \leq \frac{c_0 k}{m}$ holds for all $j\in[m]$ for some constant $c_0 = O(1)$, which gives
\begin{align*}
K := \sum_{j=1}^m \Pr\{j \in S_{\dpp}\} \leq c_0 k.
\end{align*}
By combining this result with Lemma~\ref{lem:anari_main_2}, we know that conditioned on an event that holds with probability $0.99$, Algorithm~\ref{alg:dpp_sample_ghost} produces a sample from a distribution with total variation distance $m^{-O(1)}$ from the $k$-DPP distribution with $O(\log^3 m)$ iterations of up-down walk, and size of the super set needs to satisfies $|T_i| = O(K)$. Since we know that $K = O(k)$, the constraints become $|T_i| = O(k)$ for all $i$. By Lemma~\ref{lem:sample_log}, in order to satisfy $|T_i| = t-k = O(k)$ for all $i$
holds with probability $1 - \frac{\delta}{2}$, we need to set
\begin{align*}
\mathrm{call} = O(t-k + \log\frac{2\mathrm{iter}}{\delta}) = O(k+ \log\log m + \log\frac{1}{\delta}) = O(k+\log\frac{1}{\delta}).
\end{align*}
We set $\delta = m^{-O(1)}$ and have $\log \frac{1}{\delta} = O(\log m)$, thus there exists a constant $C=O(1)$ such that it suffices to set $\mathrm{call} \geq Ck$. Then we can apply a union bound and know that with probability $1-\frac{\delta}{2}$, we can obtain all $\{T_i\}_{i=1}^{\mathrm{iter}}$ such that $|T_i| = t-k$ and $T_i \cap S_{i-1} = \varnothing$. Finally, by using Lemma~\ref{lem:tv_distance} and applying another union bound again, we know that $\Pr\{S_{\mathrm{iter}} = S_{\dpp}\} = 1 - \delta$.
\end{proof}

\paragraph{Completing the coupling argument.} We now use the guarantee for the ghost algorithm from Lemma \ref{lem:ghost_coupling} to complete the proof of Lemma \ref{lem:main_coupling} from the beginning of this section.
\begin{proof}[Proof of Lemma~\ref{lem:main_coupling}]
In Algorithm~\ref{alg:dpp_sample_real}, in order to calculate $\tau$ which is the number of calls to the oracle we need, we look at the coupled Algorithm~\ref{alg:dpp_sample_ghost}. There are two parts that need to be computed: (i) $\mathrm{call}$, the number of calls we need to obtain each $T_i$ and (ii) $\mathrm{iter}$, the number of iterations we need. Let $\delta= m^{-O(1)}$ be the failure probability. We set $\mathrm{call} \geq Ck$ for some constant $C = O(1)$ and $\mathrm{iter} = O(\log^3 m)$, then we have $\tau = \mathrm{call} \times \mathrm{iter} \geq Ck\log^3 m$.
According to Lemma~\ref{lem:ghost_coupling}, we know $S_{\mathrm{iter}} = S_{\dpp}$ holds with probability $1-\delta$. Furthermore, since we almost surely have
\begin{align*}
S_{\mathrm{iter}} \subseteq  \bigcup_{i=0}^{\mathrm{iter}-1} (S_i \cup T_i) = \tilde{S},
\end{align*}
thus we know $S_{\dpp} \subseteq \tilde{S}$ holds with probability $1-\delta$. Thus we finish the proof.
\end{proof}

Now with Lemma~\ref{lem:improve_bound}, Lemma~\ref{lem:main_coupling} and Lemma~\ref{lem:project_order}, we prove our main result as follows.

\subsection{Proof of Lemma~\ref{lem:main_sampling}}
\label{sec:proof_sample}
\begin{proof}[Proof of Lemma~\ref{lem:main_sampling}]
For $k\geq \log m$, we set $\tau \geq Ck\log^3 m$ for some constant $C = O(1)$ which is specified in the proof of Lemma~\ref{lem:ghost_coupling} and larger than $2$. We set $k' = \min\{\tau /\log^3 m, r\}$ and consider $\dpp$ sample $S_{\dpp}\sim k'$-$\mathrm{DPP}(\bar{\A}\bar{\A}^\top)$, and let $\P_{S_{\dpp}}$ be the corresponding projection matrix. Let $\tilde{S}$ be the output of Algorithm~\ref{alg:dpp_sample_real} with given $\tau$, and let $\P_{\tilde{S}}$ be the corresponding projection matrix. We define $\mathcal{E}$ as the event $\{S_{\dpp} \subseteq \tilde{S}\}$, then according to Lemma~\ref{lem:main_coupling}, event $\mathcal{E}$ holds with probability $1 - \delta$. Thus we have
\begin{align*}
\E[\P_{\tilde{S}}] = & ~ \E[\P_{\tilde{S}}\mid \mathcal{E}] \Pr\{\mathcal{E}\} + \E[\P_{\tilde{S}}\mid \mathcal{E}^c] \Pr\{\mathcal{E}^c\} \\
\succeq & ~\E[\P_{\tilde{S}}\mid \mathcal{E}] \Pr\{\mathcal{E}\} \\
\succeq & ~ \E[\P_{S_{\dpp}}\mid \mathcal{E}] \Pr\{\mathcal{E}\} \\
= & ~ \E[\P_{S_{\dpp}}] - \E[\P_{S_{\dpp}}\mid \mathcal{E}^c] \Pr\{\mathcal{E}^c\} \\
\succeq & ~ \E[\P_{S_{\dpp}}] - \Pr\{\mathcal{E}^c\}\cdot \I \\
= & ~ \E[\P_{S_{\dpp}}] - \delta \cdot \I
\end{align*}
where the first step follows from the law of total expectation, the second step follows since the projection matrices are PSD, the third step follows from the definition of $\mathcal{E}$ and Lemma~\ref{lem:project_order}, the fourth step follows from the law of total expectation, the fifth step follows from Lemma~\ref{lem:project_order} and the last step follows since $\Pr\{\mathcal{E}\} \geq 1- \delta$. Notice that since $\Q^\top\Q = \I$, we know that $\bar{\A} = \Q\A$ has the same spectrum as $\A$ and thus $\bar{\kappa}_{k}(\bar{\A})=\bar{\kappa}_{k}(\A)=\bar{\kappa}_{k}$. By applying Lemma~\ref{lem:improve_bound} to $\bar{\A}$ we know that $\lambda_{\min}^+(\E[\P_{S_{\dpp}}]) \geq \frac{k'-k}{k'-k-1 + (r-k)\bar{\kappa}_{k}^2}$. By combining this with the above result we have
\begin{align}\label{eq:S_lower_bound}
\lambda_{\min}(\E[\P_{\tilde{S}}]) \geq & ~ \lambda_{\min}^+(\E[\P_{S_{\dpp}}]) - \delta \nonumber \\
\geq & ~ \frac{k'-k}{k'-k-1 + (r-k)\bar{\kappa}_{k}^2} - \delta \nonumber \\
> & ~ \frac{\tau - k\log^3 m}{\tau-k\log^3 m + (r-k)\bar{\kappa}_{k}^2\log^3 m} - \delta \nonumber \\
= & ~ \frac{\tau - k\log^3 m}{\tau-(\bar{\kappa}_{k}^2+1)k\log^3 m + r\bar{\kappa}_{k}^2\log^3 m} - \delta
\end{align}
where the third step follows from $k' \leq \tau /\log^3 m$. We can verify that by treating $\bar{\kappa}_{k}$ as a constant and considering the above expression (Eq.\eqref{eq:S_lower_bound}) as a function of $k$, then this expression is decreasing in $k$. By using $k\leq \tau/(C\log^3 m)$, we have
\begin{align*}
\lambda_{\min}(\E[\P_{\tilde{S}}]) \geq & ~ \frac{(C-1)\tau}{C\tau-(\bar{\kappa}_{k}^2+1)\tau + Cr\bar{\kappa}_{k}^2\log^3 m} - \delta \\
= & ~ \frac{(C-1)\tau / r}{(C-1-\bar{\kappa}_k^2)\tau / r + C\bar{\kappa}_k^2\log^3 m} - \delta.
\end{align*}
We consider two cases here. If $\bar{\kappa}_k^2 \geq C-1$, then by using $C\geq 2$ we have
\begin{align*}
\lambda_{\min}(\E[\P_{\tilde{S}}]) \geq \frac{(C-1)\tau / r}{C\bar{\kappa}_k^2\log^3 m} - \delta \geq \frac{\tau / r}{2\bar{\kappa}_k^2\log^3 m} - \delta.
\end{align*}
If $\bar{\kappa}_k^2 < C-1$, then we use the assumption that $\tau / \log^3 m \leq r$ (otherwise we can simply choose $k' = r$ which is the trivial case) and have
\begin{align*}
\lambda_{\min}(\E[\P_{\tilde{S}}]) \geq & ~ \frac{(C-1)\tau / r}{(C-1-\bar{\kappa}_k^2)\log^3 m + C\bar{\kappa}_k^2\log^3 m} - \delta \\
= & ~ \frac{(C-1)\tau / r}{(C-1)\log^3 m + (C-1)\bar{\kappa}_k^2\log^3 m} - \delta \\
= & ~ \frac{\tau / r}{(\bar{\kappa}_k^2+1)\log^3 m} - \delta
\geq \frac{\tau / r}{2\bar{\kappa}_k^2\log^3 m} - \delta
\end{align*}
where we use the fact that $\bar{\kappa}_k\geq 1$. Finally, according to Remark~\ref{rem:anari_constant} we can see that the failure probability is actually $\delta = m^{-9}$. After some calculation we know that as long as $\bar{\kappa}_k \leq m^4$, we have the following holds:
\begin{align*}
\lambda_{\min}^+(\E[\P_{\tilde{S}}]) \geq \frac{\tau / r}{2\bar{\kappa}_k^2\log^3 m} - \frac{1}{m^9} \geq \frac{\tau / r}{2.5\bar{\kappa}_k^2\log^3 m}.
\end{align*}
\end{proof}

\subsection{Comparison with Direct $k$-DPP Sampling}\label{sec:compare_sample_dpp}
In this section we compare our result with the results from \cite{alv22}, and show an improvement of at least $O(k^{\omega-2})$. The main result of approximately sampling $k$-DPP in \cite{alv22} is Lemma~\ref{lem:anari_main_1}, however our setting is different in that we are not given the PSD matrix $\A\A^\top$, instead we only know $\A\in\R^{m\times n}$. We can naively first compute $\A\A^\top$ which takes $O(m^{\omega-1} n)$ then use Lemma~\ref{lem:anari_main_1}, resulting in prohibitive $\tilde{O}(m^{\omega-1}n + mk^{\omega-1})$ cost of sampling one $k$-DPP. As an alternative, we can improve this result by not computing $\A\A^\top$ explicitly. 

By slightly revising the proof in \cite{alv22}, we show that given $\A\in\R^{m \times n}$, we can obtain $s$ independent approximate samples from $k$-$\dpp(\A\A^\top)$ in time $\tilde{O}(mnk^{\omega-2} + s n k^{\omega-1})$. To be specific there are two steps, here we give a brief proof sketch of each step without rigorous proof.
\begin{itemize}
    \item Preprocessing: according to \cite{alv22}, we first need to obtain an exact $k$-DPP sample defined on an $8k \times 8k$ PSD matrix. Constructing this matrix takes $O(nk^{\omega-1})$, and sampling $k$-DPP defined by it takes $\tilde{O}(k^\omega)$, thus it takes $\tilde{O}(nk^{\omega-1} + k^\omega) = \tilde{O}(nk^{\omega-1})$. Then by using Theorem 34 in \cite{alv22}, Algorithm~\ref{alg:dpp_sample} produces marginal overestimates $q_i \geq \Pr\{i\in S_{\dpp}\}$ with sum $\sum_{i\in[m]} q_i \leq 4k$ in time $\tilde{O}(\frac{m}{k} \cdot nk^{\omega-1}) = \tilde{O}(mnk^{\omega-2})$. After the preprocessing step, the distribution is in approximately isotropic position.
    \item Sampling: after the preprocessing step, we obtain an approximately isotropic distribution, and the next step is to get independent approximate samples. This step is relatively easy, we can directly use Lemma~\ref{lem:anari_main_2} and know that we can obtain approximate samples from $k$-$\dpp(\A\A^\top)$ in time $\tilde{O}(1) \cdot \tilde{O} (n k^{\omega-1}) = \tilde{O}(nk^{\omega-1})$. Here we use the fact that $\mathcal{T}_\mu(4k, k) = \tilde{O}(n(4k)^{\omega-1} + (4k)^{\omega}) = \tilde{O} (n k^{\omega-1})$ with $\mu=k$-$\dpp(\A\A^\top)$ is the time of sampling exact DPP.
\end{itemize}
By the above two steps, we conclude that we can obtain independent approximate $k$-DPP samples in time $\tilde{O}(mnk^{\omega-2} + sn k^{\omega-1})$. In comparison, our Lemma~\ref{lem:main_coupling} only needs an $\tilde{O}(mn)$ preprocessing time and $O(k\log^3 m)$ calls to an uniform sampling oracle to guarantee the same failure probability, thus the total cost is only $\tilde{O}(mn + sk\log^3 m) = \tilde{O}(mn+sk)$ which is at least an $O(k^{\omega-2})$ improvement.

\section{Approximate Sketch-and-Project}\label{sec:approx_sketch_proj}
In this section we detail the idea and discussion in Section~\ref{sec:analysis_3}, and show how we approximate the update step of sketch-and-project by solving a preconditioned normal equation as a sub-problem. For notation convenience we use $\A\in\R^{m\times n}$ and $\b\in\R^m$ to denote the ones after randomized Hadamard transform, which is the original $\Q\A$ and $\Q\b$ respectively for $\Q = \frac{1}{\sqrt{m}}\H\D$. Specifically, $\A$ in this section corresponds to $\bar{\A}$ in Section~\ref{sec:sampling}.
\subsection{Approximate Update by Solving Sub-Problem}
In the update rule of sketch-and-project (see Eq.\eqref{eq:sketch_project_rec}), we need to compute $\w_t^* := (\I_{\tilde S}\A)^\dagger(\I_{\tilde S}\A\x_t - \I_{\tilde S}\b)$ in each iteration. We can directly compute this quantity: we first compute $\I_{\tilde S}\A$ in $O(n \tau)$, then compute $\I_{\tilde S}\A \cdot (\I_{\tilde S}\A)^\top$ in $O(n\tau^{\omega-1})$, then compute $(\I_{\tilde S}\A\A^\top\I_{\tilde S}^\top)^{\dagger}$ in $O(\tau^\omega)$, and finally apply $\I_{\tilde S}\A$ then $(\I_{\tilde S}\A\A^\top\I_{\tilde S}^\top)^{\dagger}$ to a vector, which takes $O(n\tau + \tau^2)$ time. These steps take $O(n\tau^{\omega-1} + \tau^\omega)$ in total, which is prohibitive for our per-iteration cost. Instead of directly computing $\w_t^*$, we show that it is possible to \emph{approximately} compute this quantity in $\tilde{O}(n\tau+\tau^\omega)$ time, by relying on the fact that it can be transformed into a highly under-determined least squares problem, and using a variant of the sketch-to-precondition approach \cite{rokhlin2008fast,msm14}.

Denote $\tilde{\A} = \I_{\tilde S}\A \in \R^{\tau \times n}$ and $\tilde{\b}_t = \I_{\tilde S}\A\x_t - \I_{\tilde S}\b$, notice that since $\tau \leq n$, we can transform the problem of computing $(\I_{\tilde S}\A)^\dagger(\I_{\tilde S}\A\x_t - \I_{\tilde S}\b)$ to solving an under-determined least squares problem:
\begin{align}\label{eq:under_deter_w}
\w_t^* = \arg\min \|\w\|_2 ~~~\text{subject to}~~~ \tilde{\A} \w = \tilde{\b}_t
\end{align}
which has solution $\w_t^* = \tilde{\A}^\dagger\tilde{\b}_t$. It is equivalent to solving the second kind normal equation:
\begin{align*}
\tilde{\A}\tilde{\A}^\top \u = \tilde{\b}_t, ~~~ \w = \tilde{\A}^\top \u.
\end{align*}
\begin{algorithm}[!ht]
\caption{Construct preconditioner for solving normal equation.}
\label{alg:preconditioner}
\begin{algorithmic}[1]
\State \textbf{Input: }matrix $\tilde{\A}\in\R^{\tau\times n}$;
\State Compute $\mathbf{\Phi}\tilde{\A}^\top$ by calling Lemma~\ref{lem:precondition_sparse}; \Comment{Sketch: takes $\tilde{O}(n\tau)$ time.}
\State Compute compact size SVD as $\mathbf{\Phi}\tilde{\A}^\top = \tilde{\U} \tilde{\mathbf{\Sigma}} \tilde{\V}^\top$; \Comment{Takes $O(\tau^\omega)$ time.}
\State Construct preconditioner $\mathbf{M} \leftarrow \tilde{\V} \tilde{\mathbf{\Sigma}}^{-1}$;
\\
\Return $\M$;
\end{algorithmic}
\end{algorithm}

Notice that $\tilde{\A}\tilde{\A}^\top \u = \tilde{\b}_t$ is a positive semidefinite linear system over $\u$, thus we can use CG to solve it. However since we cannot bound $\kappa(\tilde{\A})$, we need to precondition this linear system first. Notice that the subspace embedding result of Lemma~\ref{lem:precondition_sparse} directly implies a way to construct a preconditioner. By calling Algorithm~\ref{alg:preconditioner} and treating $\epsilon$ in Lemma~\ref{lem:precondition_sparse} as a fixed constant, we can obtain an invertible matrix $\M \in \R^{\tau \times \tau}$ such that $\tilde{\kappa}:=\kappa(\tilde{\A}^\top \M) = O(1)$. Thus we can precondition Problem\eqref{eq:under_deter_w} as
\begin{align}\label{eq:under_deter_w_pre}
\tilde{\w}^* = \arg\min \|\w\|_2 ~~~\text{subject to}~~~ \M^\top\tilde{\A} \w = \M^\top\tilde{\b}_t
\end{align}
and similarly, we can reformulate Problem\eqref{eq:under_deter_w_pre} to the following second kind normal equation:
\begin{align*}
(\M^\top \tilde{\A}\tilde{\A}^\top \M) \u = \M^\top\tilde{\b}_t, ~~~ \w = \tilde{\A}^\top \M\u.
\end{align*}
We use CG to solve this positive semidefinite linear system over $\u$, and have the following lemma.

\begin{lemma}[Convergence]\label{lem:approx_converge}
Given $\tilde{\A} \in \R^{\tau \times n}, \tilde{\b} \in \R^\tau$, let $\w^* = \tilde{\A}^\dagger\tilde{\b}$ be the least squares solution of $\tilde{\A}\w = \tilde{\b}$. Given $\delta <1/2$, if the subspace embedding matrix $\mathbf{\Phi} \in \R^{\phi\times n}$ from Lemma \ref{lem:precondition_sparse} satisfies $\phi = O(\tau+\log(1/\delta))$, then with probability $1 - \delta$, Algorithm~\ref{alg:preconditioner} outputs $\M$ in time $O(n\tau \log(\tau/\delta) + \tau^{\omega})$ such that $\tilde{\kappa} = \kappa(\tilde{\A}^\top\M) = O(1)$, and solving the following system
\begin{align}\label{eq:equi_cg}
(\M^\top \tilde{\A}\tilde{\A}^\top \M) \u = \M^\top\tilde{\b}, ~~~ \w = \tilde{\A}^\top \M\u
\end{align}
by using CG and starting from $\w_0 = \mathbf{0}$ will yield the following convergence result:
\begin{align*}
\|\w_s - \w^*\| \leq 2\left(\frac{\tilde{\kappa}-1}{\tilde{\kappa}+1}\right)^{s} \cdot \|\w^*\|.
\end{align*}
\end{lemma}
\begin{proof}
Let $\u^*$ be the min-length solution of the PSD linear system $(\M^\top \tilde{\A}\tilde{\A}^\top \M) \u = \M^\top\tilde{\b}$. Denote $\tilde{\w}^*$ as the solution of Problem\eqref{eq:equi_cg}, then we have
\begin{align*}
\tilde{\w}^* = & ~  \tilde{\A}^\top \M\u^* \\
= & ~  \tilde{\A}^\top \M (\M^\top\tilde{\A}\tilde{\A}^\top\M)^\dagger\M^\top\tilde{\b} \\
= & ~ \tilde{\A}^\top \M \M^{-1} (\tilde{\A} \tilde{\A}^\top)^{\dagger}(\M^\top)^{-1}\M^\top\tilde{\b} \\
= & ~ \tilde{\A}^\top (\tilde{\A}\tilde{\A}^\top)^{\dagger}\tilde{\b} \\
= & ~ \tilde{\A}^\dagger \tilde{\b} = \w^*
\end{align*}
which shows that the optimal solution of Problem\eqref{eq:equi_cg} is the same as that of $\tilde{\A}\w = \tilde{\b}$. According to Lemma~\ref{lem:precondition_sparse} we know that with the choice of $\phi= O(\tau + \log(1/\delta))$, $\tilde{\kappa} = O(1)$ holds with probability $1 - \delta$. According to Lemma~\ref{lem:cg}, the convergence guarantee is
\begin{align*}
\|\tilde{\A}^\top\M(\u_s - \u^*)\| \leq 2 \left(\frac{\tilde{\kappa} - 1}{\tilde{\kappa}+1} \right)^s \cdot \|\tilde{\A}^\top\M \u^*\|
\end{align*}
and by doing the change of variable back from $\u$ to $\w$ we have
\begin{align*}
\|\w_s - \w^*\| \leq 2\left(\frac{\tilde{\kappa} - 1}{\tilde{\kappa} + 1}\right)^s \cdot\|\w^*\|.
\end{align*}
Finally we compute the cost of constructing the preconditioner. Following Lemma~\ref{lem:precondition_sparse} and Algorithm~\ref{alg:preconditioner}, we first construct $\mathbf{\Phi}\tilde{\A}^\top$ which takes $O(n\tau \log(\tau/\delta) + \tau^2\log^4(\tau/\delta))$, then compute the compact size SVD as $\mathbf{\Phi}\tilde{\A}^\top = \tilde{\U}\tilde{\mathbf{\Sigma}} \tilde{\V}^\top$ which takes $O(\phi\tau^{\omega-1})$ where $\phi = O(\tau + \log(1 / \delta))$, and finally we set $\M = \V\tilde{\mathbf{\Sigma}}^{-1}$ which takes $O(\tau^2)$. Thus we conclude that the total time complexity is
\begin{align*}
O(n\tau \log(\tau/\delta) + \tau^2\log^4(\tau/\delta) + \tau^{\omega-1}\phi) = O(n\tau \log(\tau/\delta) + \tau^{\omega}) = \tilde{O}(n\tau + \tau^{\omega}).
\end{align*}
\end{proof}

\subsection{Proof of Lemma~\ref{lem:main_convergence}}
\label{sec:proof_convergence}
We are now ready to provide the convergence analysis for sketch-and-project with an approximate update, as it is implemented in Algorithm \ref{alg:main}.
\label{sec:proof_approx_project}
\begin{proof}[Proof of Lemma~\ref{lem:main_convergence}]
In outer iteration $t$, we analyze the update as follows.
\begin{align*}
\E\|\x_{t+1} - \x^*\|^2 = & ~ \E\|\x_t - \w_{s_{\max}} - \x^*\|^2 \\
= & ~ \E\|(\x_t - \x^* - \w_t^*) + (\w_t^* - \w_{s_{\max}})\|^2 \\
= & ~ \E\|\x_t - \x^* - \w_t^*\|^2 + \|\w_t^* - \w_{s_{\max}}\|^2 + 2\E[(\x_t - \x^* - \w_t^*)^\top(\w_t^* - \w_{s_{\max}})] \\
\leq & ~ \E\|\x_t - \x^* - \w_t^*\|^2 + \|\w_t^* - \w_{s_{\max}}\|^2 + 2\|\w_t^* - \w_{s_{\max}}\| \cdot (\E\|\x_t - \x^* - \w_t^*\|^2)^{1/2}.
\end{align*}
For the first part, as we already discussed in Section \ref{sec:analysis_1}, we have
\begin{align}
\label{eq:null_space}
\E\|\x_t - \x^* - \w_t^*\|^2 = & ~ (\x_t - \x^*)^\top\E[(\I-\P_{{\tilde S}})](\x_t - \x^*) \nonumber \\
\leq & ~ (1 - \lambda_{\min}^+(\E[\P_{{\tilde S}}])) \cdot \|\x_t - \x^*\|^2.
\end{align}
This is due to the update rule of sketch-and-project:
\begin{align*}
\x_{t+1} = \x_t - \w_t^* = \x_t - \A^\top\I_{\tilde S}^\top(\I_{\tilde S}\A\A^\top\I_{\tilde S}^\top)^\dagger (\I_{\tilde S}\A\x_t - \I_{\tilde S}\b)
\end{align*}
and the fact that if $\x_t-\x^*$ is not in the null-space of $\E[\P_{\tilde S}]$, then neither will be $\x_{t+1}-\x^*$. This fact follows because
\begin{align*}
\E[\P_{\tilde S}] = & ~ \A^\top \E[\I_{\tilde S}^\top (\I_{\tilde S}\A\A^\top\I_{\tilde S}^\top)^\dagger\I_{\tilde S}] \A 
\end{align*}
which implies that the nullspace of $\E[\P_{\tilde S}]$ contains the nullspace of $\A^\top\A$. At the same time, in the proof of Lemma~\ref{lem:main_sampling} we showed that $\E[\P_{\tilde S}]\succeq \E[\P_{S_{\dpp}}]-\delta\I\succeq \frac12\E[\P_{S_{\dpp}}]$, where $S_{\dpp}\sim k'$-DPP$(\A\A^\top)$, so the nullspace of $\E[\P_{\tilde S}]$ is no larger than that of $\E[\P_{S_{\dpp}}]$. But we know that the nullspace of $\E[\P_{S_{\dpp}}]$ is exactly the same as that of $\A^\top\A$ using Eq.\eqref{eq:eigen_decompose}.
Finally, denoting the singular value decomposition of $\A$ as $\A = \U\mathbf{\Sigma}\V^\top$ where $\U \in \R^{m \times m}, \mathbf{\Sigma} \in \R^{m\times n}$ and $\V \in \R^{n\times n}$, the update rule can be seen as $\x_{t+1} = \x_t - \V\cdot \v_t$ for some vector $\v_t$, thus if $\x_t - \x^*$ is not in the null-space of $\E[\P_{\tilde S}]$, then $\x_{t+1}-\x^*$ will never be in the null-space of $\E[\P_{\tilde S}]$. In particular, if we set $\x_0 = \mathbf{0}$, then notice that $\x^* = \A^\dagger \b = \V\mathbf{\Sigma}^\dagger \U^\top\b$, we know $\x_0 - \x^* = \x^*$ is not in the nullspace of $\E[\P_{\tilde S}]$. Thus we know Eq.\eqref{eq:null_space} holds.

For the second and third part, denote $\Delta_t := \|\w_t^* - \w_{s_{\max}}\|$ and $\rho := \lambda_{\min}^+(\E[\P_{\tilde S}])$. Then, we have
\begin{align}\label{eq:rho_delta}
\E\|\x_{t+1} - \x^*\|^2 \leq & ~ (1 - \lambda_{\min}^+(\E[\P_{\tilde S}])) \cdot \|\x_t - \x^*\|^2 + 2\Delta_t\sqrt{1 - \lambda_{\min}^+(\E[\P_{\tilde S}])}\cdot \|\x_t - \x^*\| + \Delta_t^2 \nonumber \\
= & ~ (\sqrt{1 - \lambda_{\min}^+(\E[\P_{\tilde S}])} \cdot \|\x_t - \x^*\| + \Delta_t)^2 \nonumber \\
= & ~ (\sqrt{1 - \rho} \cdot \|\x_t - \x^*\| + \Delta_t)^2.
\end{align}
We define $\mathcal{E}$ as the event that $\kappa(\tilde{\A}^\top \M) = O(1)$, then according to Lemma~\ref{lem:precondition_sparse}, as long as we set $\phi = O(\tau+\log(1 / \delta_1))$, then event $\mathcal{E}$ holds with probability $1 - \delta_1$. Denote $\tilde{\x}_{t+1} = \x_t - \w_t^*$ be the update of \emph{exact} sketch-and-project at iteration $t$, then we have
\begin{align*}
\|\w_t^*\| = \|\tilde{\x}_{t+1} - \x_t\| \leq \|\tilde{\x}_{t+1} - \x^*\| + \|\x_t - \x^*\| \leq 2 \|\x_t - \x^*\|.
\end{align*}
If event $\mathcal{E}$ holds, then we choose $\epsilon = \frac{\rho}{16}$ and the number of iterations as $s_{\max} = O(\tilde{\kappa}\log(1/\epsilon)) = O(\log(1/\rho))$, according to Lemma~\ref{lem:approx_converge} we have
\begin{align*}
\Delta_t = \|\w_{s_{\max}} - \w_t^*\| \leq \frac{\rho}{8} \cdot \|\w_t^*\| \leq \frac{\rho\cdot \|\x_t - \x^*\|}{4}.
\end{align*}
If event $\mathcal{E}$ does not hold, we denote $\tilde{\w}^* = (\M^\top \tilde{\A})^\dagger \M^\top \tilde{\b}_t$ as the solution of Problem \eqref{eq:equi_cg}, where in comparison we have $\w_t^* = \tilde{\A}^\top \tilde{\b}_t$ is the solution of the original Problem \eqref{eq:under_deter_w}. Denote $W^* = \{\w \mid \tilde{\A}\w = \tilde{\b}_t\}$ and $\tilde{W}^* = \{\w \mid \M^\top\tilde{\A}\w = \M^\top\tilde{\b}_t\}$, notice that $W^* \subseteq \tilde{W}^*$, thus we have $\|\tilde{\w}^*\| = \min_{\w \in \tilde{W}^*} \|\w\| \leq \min_{\w \in W^*} \|\w\| = \|\w_t^*\|$. Then we have
\begin{align*}
\Delta_t = & ~ \|\w_t^* - \w_{s_{\max}}\| \\
\leq & ~ \|\w_{s_{\max}} - \tilde{\w}^*\| + \|\tilde{\w}^* - \w_t^* \| \\
\leq & ~ 2\Big(\frac{\tilde{\kappa}-1}{\tilde{\kappa}+1}\Big)^{s_{\max}} \cdot \|\tilde{\w}^*\| + \|\tilde{\w}^* - \w_t^* \| \\
\leq & ~ 2\|\tilde{\w}^*\| + \|\tilde{\w}^*\| + \|\w_t^*\| \\
\leq & ~ 4\|\w_t^*\| \\
\leq & ~ 8 \|\x_t - \x^*\|
\end{align*}
where we use Lemma~\ref{lem:approx_converge} in the third step. By combining above two results with Eq.\eqref{eq:rho_delta} we have 
\begin{align*}
\E\|\x_{t+1} - \x^*\|^2 = & ~ \E[\|\x_{t+1} - \x^*\|^2 \mid \mathcal{E}]\cdot \Pr\{\mathcal{E}\} + \E[\|\x_{t+1} - \x^*\|^2 \mid \bar{\mathcal{E}}]\cdot \Pr\{\bar{\mathcal{E}}\} \\
\leq & ~ (1-\delta_1) \cdot (\sqrt{1 - \rho} \cdot \|\x_t - \x^*\| + \frac{\rho}{4}\cdot\|\x_t - \x^*\|)^2 + \delta_1\cdot (\sqrt{1 - \rho} \cdot \|\x_t - \x^*\| + 8\|\x_t - \x^*\|)^2 \\
\leq & ~ (\sqrt{1 - \rho} \cdot \|\x_t - \x^*\| + \frac{\rho}{4}\cdot\|\x_t - \x^*\|)^2 + \delta_1\cdot (9 \|\x_t - \x^*\|)^2 \\
\leq & ~ (1 - \frac{\rho}{2} + 81\delta_1)\cdot \|\x_t - \x^*\|^2
\end{align*}
where we use Eq.\eqref{eq:rho_delta} and the fact that $\sqrt{1 - \rho} + \rho / 4 \leq \sqrt{1 - \rho / 2}$. We choose $\delta_1 = \rho / 324$ and have
\begin{align*}
\E\|\x_{t+1} - \x^*\|^2 \leq \left(1 - \frac{\rho}{4}\right)\cdot \|\x_t - \x^*\|^2.
\end{align*}
\end{proof}

\section{Positive Semidefinite Linear Systems}
\label{sec:psd}
In this section we consider solving a positive semidefinite linear system $\A\x = \b$ where $\A\in\R^{n \times n}$ is a PSD matrix with $\rank(\A) = r$ and $\b\in\R^n$. Denote $\{\lambda_i\}_{i = 1}^n$ as the eigenvalues of $\A$ in decreasing order and define $\bar{\kappa}_k := \frac{1}{r-k}\sum_{j > k}^r \lambda_j / \lambda_r$. Notice that the definition of $\bar{\kappa}_k$ is slightly different from the rectangular matrix case in the sense that we are averaging the tail ratios of eigenvalues. Similar to the sketch-and-project analysis of the rectangular matrix setting, we need to analyze the quantity $\lambda_{\min}^+(\E[\P_S]) = \lambda_{\min}^+(\A^{1/2}\E[(\I_S \A \I_S)^\dagger] \A^{1/2})$. The difference here is that we naturally have access to this PSD matrix $\A$, thus, we can directly use Lemma~\ref{lem:anari_main_1} for approximately sampling $\tau$-DPPs. Furthermore, because of the PSD structure, instead of a Kaczmarz-type row-sampling algorithm, we use a coordinate descent-type algorithm. Notice that in this section we use $\tau$ to denote the $\dpp$ sample size, which is not the same as the uniform sample size in previous sections.

\begin{algorithm}[!ht]
\caption{Fast coordinate descent-type solver for PSD linear systems.}
\label{alg:main_psd}
\begin{algorithmic}[1]
\State \textbf{Input: }matrix $\A\in\R^{n\times n}$, vector $\b \in \R^n$, $\dpp$ sample size $\tau$, iterate $\x_0$, outer iteration $t_{\max}$;
\State Precompute the marginal overestimates as in Theorem 4 of \cite{alv22}; \Comment{Takes $\tilde{O}(n\tau^{\omega-1})$ time.}
\For{$t = 0, 1, \ldots, t_{\max}-1$}
\State Generate approximate $\tau$-$\dpp(\A)$ sample $S$ using Lemma~\ref{lem:anari_main_2};
\Comment{Takes $\tilde{O}(\tau^\omega)$ time.}
\State Compute $\A_{S, S}\leftarrow \I_S \A \I_S^\top$; \Comment{Takes $O(\tau^2)$ time.}
\State Compute $\tilde{\b} \leftarrow \I_S\A\x_t - \I_S\b$; \Comment{Takes $O(\nnz(\A_{S,:}))$ time.}
\State Compute $\w_t \leftarrow \A_{S, S}^\dagger\tilde{\b}$; \Comment{Takes $O(\tau^{\omega})$ time.}
\State Update $\x_{t+1} = \x_t - \I_S^\top\w_t$; \Comment{Takes $O(\tau)$ time.}
\EndFor \\
\Return $\tilde{\x} = \x_{t_{\max}}$;
\end{algorithmic}
\end{algorithm}

\begin{theorem}[Fast solver for PSD linear systems]\label{thm:main_psd}
Given PSD matrix $\A\in\R^{n \times n}$ with $\rank(\A) = r$ and $\b \in \R^n$, let $\{\lambda_i\}_{i=1}^n$ be the eigenvalues of $\A$ in decreasing order. For $\log n \leq k < r$, define $\bar{\kappa}_k := \frac{1}{r-k}\sum_{j > k}^r \lambda_j / \lambda_r$ and assume for simplicity that $\bar{\kappa}_k \leq n^8$.
Let $\x^* = \A^\dagger\b$ be the min-length solution of the linear system $\A\x = \b$. Assume $\x_0 - \x^*$ is not in the nullspace of $\A$. Then conditioned on an event that holds with probability $0.99$ (and only depends on the preprocessing step in line 2 of Algorithm~\ref{alg:main_psd}), for some $C=O(1)$, running Algorithm~\ref{alg:main_psd} with choice $\tau = \min\{Ck, r\}$ and $\x_0$ will yield the following convergence result:
\begin{align*}
\E\|\x_t - \x^*\|_{\A}^2 \leq \left(1 - \frac{\tau/r}{3\bar{\kappa}_k}\right)^t \cdot \|\x_0 - \x^*\|_{\A}^2.
\end{align*}
By further choosing $\x_0 = \mathbf{0}$ and $t_{\max} =  O(r \bar{\kappa}_k/ \tau \cdot\log(1/\epsilon))$ for given $\epsilon > 0$, Algorithm~\ref{alg:main_psd} outputs $\tilde{\x}$ such that $\|\tilde{\x} - \x^*\|_{\A}^2 \leq \epsilon \|\x^*\|_{\A}^2$ holds with probability $0.98$ in time:
\begin{align*}
O\Big((\nnz^*(\A) + n\tau^{\omega-1} \log^4 n)\cdot \bar{\kappa}_k\log1/\epsilon\Big).
\end{align*}
where $\nnz^*(\A)=n\cdot\max_i\{\nnz(\A_{i,:})\}\leq n^2$.
\end{theorem}

\begin{proof}
In the case of solving PSD linear system, instead of using approximated sketch-and-project, we use exact block coordinate descent with the following update rule
\begin{align*}
\x_{t+1} = \x_t - \I_S^\top(\A_{S, S})^\dagger (\A\x_t - \b)_{S} = \x_t - \I_S^\top(\I_S\A\I_S^\top)^\dagger \I_S(\A\x_t - \b)
\end{align*}
where $S$ is an approximated $\tau$-DPP sample obtained by Lemma~\ref{lem:anari_main_2}. Here for notation simplicity we use $S$, but the sample is different and independent at each iteration $t$. Notice that $\x^* = \A^\dagger\b$ is the optimal solution, similar to the proof of Lemma~\ref{lem:main_convergence} in Section~\ref{sec:proof_approx_project}, we have
\begin{align*}
\x_{t+1} - \x^* = (\I-\I_S^\top(\I_S\A\I_S^\top)^\dagger\I_S\A) (\x_t-\x^*).
\end{align*}
Since we assume that $\x_0 - \x^*$ is not in the nullspace of $\A$, by induction we know $\x_{t} - \x^*$ will always be not in the nullspace of $\A$ (same as the argument in proof of Lemma~\ref{lem:main_convergence} in Section~\ref{sec:proof_convergence}). Next, we proceed with a standard convergence analysis of block coordinate descent:
\begin{align*}
\A^{1/2}(\x_{t+1} - \x^*) = & ~ (\A^{1/2}-\A^{1/2}\I_S^\top(\I_S\A\I_S^\top)^\dagger\I_S\A) (\x_t-\x^*) \\
= & ~ (\I-\A^{1/2}\I_S^\top(\I_S\A\I_S^\top)^\dagger\I_S\A^{1/2}) \A^{1/2}(\x_t-\x^*) \\
= & ~ (\I - \P_S)\A^{1/2}(\x_t-\x^*)
\end{align*}
where $\P_S = \A^{1/2}\I_S^\top(\I_S\A\I_S^\top)^\dagger\I_S\A^{1/2}$ is the projection matrix according to $\A^{1/2}$ and $S$.
Similarly, but with slight difference, we have
\begin{align*}
\E\|\x_{t+1} - \x^*\|_{\A}^2 = & ~ \E\|\A^{1/2}(\x_{t+1} - \x^*)\|^2 \\
= & ~ \E[(\A^{1/2}\x_{t} - \A^{1/2}\x^*)^\top (\I-\P_S)^2 (\A^{1/2}\x_{t} - \A^{1/2}\x^*)] \\
= & ~ (\A^{1/2}\x_{t} - \A^{1/2}\x^*)^\top \E[(\I-\P_S)^2](\A^{1/2}\x_{t} - \A^{1/2}\x^*) \\
= & ~ (\A^{1/2}\x_{t} - \A^{1/2}\x^*)^\top (\I - \E[\P_S])(\A^{1/2}\x_{t} - \A^{1/2}\x^*) \\
\leq & ~ (1-\lambda_{\min}^+(\E[\P_S]))\cdot \|\A^{1/2}(\x_t-\x^*)\|^2 \\
=  & ~ (1-\rho)\cdot \|(\x_t-\x^*)\|_{\A}^2
\end{align*}
where we denote $\rho = \lambda_{\min}^+(\E[\P_S])$. According to Theorem 4 of \cite{alv22} we know with probability $1-n^{-10} > 0.99$, the marginal probabilities are approximately isotropic. According to Lemma~\ref{lem:anari_main_2} and Lemma~\ref{lem:tv_distance}, conditioned on this event, there is a coupling $(S, S_{\dpp})$ where $S_{\dpp}\sim \tau$-$\dpp(\A)$, such that $S = S_{\dpp}$ with probability $1 - n^{-O(1)}$. Denote $\mathcal{E}$ as the event $\{S = S_{\dpp}\}$, we have $\Pr\{\mathcal{E}\} = 1 - \delta$ where $\delta = n^{-O(1)}$. Similar to the proof of Lemma~\ref{lem:main_sampling} in Section~\ref{sec:proof_sample}, we have
\begin{align*}
\E[\P_{S}] = & ~ \E[\P_{S}\mid \mathcal{E}] \Pr\{\mathcal{E}\} + \E[\P_{S}\mid \mathcal{E}^c] \Pr\{\mathcal{E}^c\} \\
\succeq & ~\E[\P_{S}\mid \mathcal{E}] \Pr\{\mathcal{E}\} \\
= & ~ \E[\P_{S_{\dpp}}\mid \mathcal{E}] \Pr\{\mathcal{E}\} \\
= & ~ \E[\P_{S_{\dpp}}] - \E[\P_{S_{\dpp}}\mid \mathcal{E}^c] \Pr\{\mathcal{E}^c\} \\
\succeq & ~ \E[\P_{S_{\dpp}}] - \Pr\{\mathcal{E}^c\}\cdot \I \\
= & ~ \E[\P_{S_{\dpp}}] - \delta \cdot \I.
\end{align*}
Next we apply Lemma~\ref{lem:improve_bound} to $\A^{1/2} \in \R^{n\times n}$ with rank $r$, by denoting $S_{\dpp}\sim\tau$-$\dpp(\A)$ we have
\begin{align*}
\lambda_{\min}^+(\E[\P_{S_{\dpp}}]) \geq & ~\frac{\tau-k}{\tau-k-1 + (r-k)\bar{\kappa}_{k}} \\
> & ~ \frac{(C-1)\tau}{(C-1)\tau + C(r-k)\bar{\kappa}_{k}} \\
= & ~ \frac{(C-1)\tau/r}{(C-1-\bar{\kappa}_{k})\tau/r + C\bar{\kappa}_{k}}.
\end{align*}
Similar to the proof of Lemma~\ref{lem:main_sampling}, we consider two cases here. If $\bar{\kappa}_k \geq C-1$, then by using $C\geq 2$ we have
\begin{align*}
\lambda_{\min}^+(\E[\P_{S_{\dpp}}]) \geq \frac{(C-1)\tau/r}{ C\bar{\kappa}_{k}} \geq \frac{\tau/r}{ 2\bar{\kappa}_{k}}.
\end{align*}
If $\bar{\kappa}_k^2 < C-1$, then we use the assumption that $\tau \leq r$ and have
\begin{align*}
\lambda_{\min}^+(\E[\P_{S_{\dpp}}]) \geq \frac{(C-1)\tau/r}{(C-1-\bar{\kappa}_{k}) + C\bar{\kappa}_{k}} = \frac{\tau/r}{1 + \bar{\kappa}_{k}}\geq \frac{\tau/r}{ 2\bar{\kappa}_{k}} 
\end{align*}
Finally, according to Remark~\ref{rem:anari_constant} we can see that the failure probability is actually $\delta = n^{-9}$. After some calculation we know that as long as $\bar{\kappa}_k \leq n^8$, we have the following holds:
\begin{align*}
\rho = \lambda_{\min}^+(\E[\P_S]) \geq \lambda_{\min}^+(\E[\P_{S_\dpp}]) -\delta \geq \frac{\tau/r}{ 2\bar{\kappa}_{k}} - \frac{1}{n^9}\geq \frac{\tau/r}{ 3\bar{\kappa}_{k}}.
\end{align*}
Thus we have the following convergence result:
\begin{align*}
\E\|\x_t - \x^*\|_{\A}^2 \leq \left(1 - \frac{\tau/r}{3\bar{\kappa}_{k}}\right)^t \cdot \|\x_0 - \x^*\|_{\A}^2.
\end{align*}
By further choosing $t_{\max} = O(r\bar{\kappa}_k/\tau\cdot\log(100/\epsilon))$ and denoting $\tilde{\x} = \x_{t_{\max}}$, we have $\E\|\tilde{\x} - \x^*\|_{\A}^2\leq \epsilon\|\x^*\| / 100$. By using Markov's inequality and taking an union bound, we have $\|\tilde{\x} - \x^*\|_{\A}^2\leq \epsilon\|\x^*\|_{\A}^2$ holds with probability $0.98$. Next we compute the time complexity. According to Theorem 34 in \cite{alv22}, the preprocessing step takes $O(\frac{n}{\tau}\log n\cdot \tau^{\omega} \log^3 n) = O(n\tau^{\omega-1} \log^4 n)$. As for the sampling procedure, generating each approximate $\tau$-$\dpp$ sample takes $O(\tau^{\omega} \log^3 n)$. Further notice that in Algorithm~\ref{alg:main_psd} the step of computing $\tilde{\b} = \I_S \A\x_t - \I_S\b$ takes $O(\nnz(\A_{S,:}))$, thus the overall time complexity is
\begin{align*}
& ~ O\left(n\tau^{\omega - 1}\log^4 n + r\bar{\kappa}_k /\tau\cdot\log(1/\epsilon) \cdot(\max_{S:|S|=\tau}\nnz(\A_{S,:}) + \tau^{\omega}\log^3 n)\right) \\
= & ~ O\Big((\nnz^*(\A) + n\tau^{\omega-1} \log^4 n)\cdot\bar{\kappa}_k \log1/\epsilon\Big)
\end{align*}
where $\nnz^*(\A)=n\cdot\max_i\{\nnz(\A_{i,:})\}\leq n^2$.
\end{proof}

\section{Implicit PSD Linear Systems}
\label{sec:proof_main_2}
In this section we give a proof of Theorem~\ref{thm:main_2} which is our second main theorem. Though there are some differences, the proof idea is similar with that of Theorem~\ref{thm:main}, which is detailed in Section~\ref{sec:analysis}. Here, we actually show a more general result, which is that one can use Algorithm \ref{alg:main_2} to solve implicit PSD linear systems of the form $\A^\top\A\x=\c$, with any $\c$ (this becomes useful for our sparse least squares solver in Section \ref{s:ls}). Then, following standard arguments \cite{leventhal2010randomized,gower2015randomized}, we illustrate how to instantiate this with $\c=\A^\top\b$ to solve the original system $\A\x=\b$ via normal equations.
\begin{proof}[Proof of Theorem~\ref{thm:main_2}]
In the algorithm, we preprocess the $\A^\top\A\x=\c$ linear system (from both left and right) with $\Q=\frac1{\sqrt n}\H\D$:
\begin{align*}
(\Q\A^\top)(\A\Q^\top)\z = \Q\c, ~~~\x = \Q^\top\z.
\end{align*}
For notation convenience we use $\A$ and $\c$ to denote the ones after randomized Hadamard transform (which is the original $\A\Q^\top$ and $\Q\c$), and we will transform back to the original $\A$ and $\c$ in the end of the proof. Thus we need to solve the implicit PSD linear system $\A^\top\A\z = \c$, where the optimal solution is $\z^* = (\A^\top\A)^\dagger\c$. Denote $\I_{\tilde{S}}^\top\in\R^{n \times \tau}$ as the column sampling matrix where $\tilde{S}$ is the output of Algorithm~\ref{alg:dpp_sample_real} and $\tilde{\A} = \A\I_{\tilde{S}}^\top \in \R^{m\times \tau}$. Consider the following sketch-and-project update rule:
\begin{align}\label{eq:exact_sketch_proj}
\z_{t+1} = \z_t - \I_{\tilde{S}}^\top(\I_{\tilde{S}}\A^\top\A\I_{\tilde{S}}^\top)^\dagger\I_{\tilde{S}}(\A^\top\A\z_t - \c).
\end{align}
Denote $\y_t = \A\z_t$ and $\tilde{\c} = \tilde{\A}^\top\y_t - \I_{\tilde{S}}\c$, by right multiplying by $\A$ we have
\begin{align*}
\y_{t+1} = \y_t - \tilde{\A}(\tilde{\A}^\top\tilde{\A})^\dagger(\tilde{\A}^\top\y_t - \I_{\tilde{S}}\c) = \y_t - \tilde{\A}\underbrace{(\tilde{\A}^\top\tilde{\A})^\dagger\tilde{\c}}_{\w_t^*}.
\end{align*}
Similarly, instead of directly computing $\w_t^*$ we consider solving the sub-problem $\tilde{\A}^\top\tilde{\A}\w = \tilde{\c}$. We construct the preconditioner $\M$ similarly as Lemma~\ref{lem:approx_converge}, and consider the following linear system
\begin{align*}
\M^\top\tilde{\A}^\top\tilde{\A}\M \u = \M^\top\tilde{\c}, ~~~\w = \M\u,
\end{align*}
then Lemma~\ref{lem:approx_converge} guarantees that $\tilde{\kappa} = \kappa(\tilde{\A}\M) = O(1)$ holds with probability $1-\delta_1$. Denote $\u^* = \M^{-1}(\tilde{\A}^\top\tilde{\A})^\dagger\tilde{\c}$ as the optimal solution, by calling $\mathrm{CG}(\M^\top\tilde{\A}^\top, \M^\top\tilde{\c}, s)$ and starting from $\u_0 = \mathbf{0}$, Lemma~\ref{lem:cg} gives the following convergence result:
\begin{align*}
\|\tilde{\A}\M(\u_s - \u^*)\| \leq 2\left(\frac{\tilde{\kappa}-1}{\tilde{\kappa}+1}\right)^{s} \cdot \|\tilde{\A}\M\u^*\|.
\end{align*}
where $\tilde{\kappa} = O(1)$ holds with probability $1-\delta_1$. Define $\mathcal{E}$ as the event that $\tilde{\kappa} = O(1)$, if event $\mathcal{E}$ holds, then by choosing $s_{\max} = O(\log(1/\epsilon))$ and changing the variable back from $\u$ to $\w$ we have
\begin{align*}
\|\tilde{\A}(\w_{s_{\max}} - \w_t^*)\| \leq \epsilon\|\tilde{\A}\w_t^*\|.
\end{align*}
Denote $\tilde{\z}_{t+1} = \z_t - \I_{\tilde{S}}^\top \w_t^*$ as the update of \textbf{exact} sketch-and-project, notice that our approximate update is $\z_{t+1} = \z_t - \I_{\tilde{S}}^\top\w_{s_{\max}}$. Denote $\Delta_t = \|\A(\tilde{\z}_{t+1} - \z_{t+1})\|$ as the matrix $\A$ norm difference between our approximate update and exact sketch-and-project update, and we bound it as
\begin{align*}
\Delta_t = & ~ \|\A(\tilde{\z}_{t+1} - \z_{t+1})\| 
=\|\tilde{\A} (\w_{s_{\max}} -\w_t^*)\|\\
\leq & ~ \epsilon \|\tilde{\A}\w_t^*\| 
=  \epsilon \|\A\tilde{\z}_{t+1} - \A\z_t\| \\
\leq & ~ \epsilon (\|\A\tilde{\z}_{t+1} -\A\z^*\| +\|\A\z_{t} -\A\z^*\|) 
\leq  2\epsilon\|\A(\z_{t} - \z^*)\|.
\end{align*}
Notice that $\tilde{\z}_{t+1}$ is the same as the update of exact block coordinate descent on the PSD linear system $\A^\top\A\z=\c$, so similarly as before, we have:
\begin{align*}
\E\|\A(\tilde{\z}_{t+1} - \z^*)\|^2 
\leq & ~ (1-\lambda_{\min}^+(\E[\P_S]))\cdot \| \A(\z_t-\z^*)\|^2 \\
=  & ~ (1-\rho)\cdot \| \A(\z_t-\z^*)\|^2.
\end{align*}
where we denote $\rho = \lambda_{\min}^+(\E[\P_{S}])$. By combining the above two results we have
\begin{align*}
 \E\|\A(\z_{t+1} - \z^*)\|^2 
= & ~ \E\|(\A\z_{t+1} - \A\tilde{\z}_{t+1}) + (\A\tilde{\z}_{t+1} - \A\z^*)\|^2 \\
\leq & ~ \Delta_t^2 + \E\|\A(\tilde{\z}_{t+1} - \z^*)\|^2 + 2\Delta_t \cdot (\E\|\A(\tilde{\z}_{t+1} - \z^*)\|^2)^{1/2} \\
\leq & ~ \Delta_t^2 + (1-\rho)\cdot \| \A(\z_t-\z^*)\|^2 + 2\Delta_t\sqrt{1-\rho}\cdot\E\|\A(\z_{t+1} - \tilde{\z}_{t+1})\| \\
= & ~ (\sqrt{1-\rho}\cdot \|\A(\z_t-\z^*)\| + \Delta_t)^2.
\end{align*}
Notice that if event $\mathcal{E}$ holds, then we have $\Delta_t \leq 2\epsilon \|\A(\z_t-\z^*)\|$ by choosing $s_{\max} = \log(1/\epsilon)$. If event $\mathcal{E}$ does not hold, we use the same trick as in the proof of Lemma~\ref{lem:main_convergence} and bound $\Delta_t$ as $\Delta_t \leq 8 \|\A(\z_t - \z^*)\|$. To conclude we have
\begin{align*}
\E\|\A(\z_{t+1} - \z^*)\|^2 = & ~ \E[\|\A(\z_{t+1} - \z^*)\|^2 \mid \mathcal{E}] \cdot\Pr\{\mathcal{E}\} + \E[\|\A(\z_{t+1} - \z^*)\|^2 \mid \bar{\mathcal{E}}] \cdot\Pr\{\bar{\mathcal{E}}\} \\
\leq & ~ (1-\delta_1)\cdot(\sqrt{1-\rho} +2\epsilon)^2\|\A(\z_t-\z^*)\|^2 + \delta_1\cdot (\sqrt{1-\rho} + 8)^2\|\A(\z_t-\z^*)\|^2 \\
\leq & ~ [(\sqrt{1-\rho} +2\epsilon)^2 + 81\delta_1]\cdot \|\A(\z_t-\z^*)\|^2.
\end{align*}
We set $\epsilon = \rho / 8$ and $\delta_1 = \rho/324$, notice that $\sqrt{1 - \rho} + \rho / 4 \leq \sqrt{1 - \rho / 2}$, thus we have
\begin{align*}
\E\|\A(\z_{t+1} - \z^*)\|^2 \leq \left(1 - \frac{\rho}{4}\right)\cdot \|\A(\z_t-\z^*)\|^2.
\end{align*}
By applying Lemma~\ref{lem:main_sampling} to $\A^\top$, we know that by setting $\tau \geq Ck\log^3 n$ for some constant $C = O(1)$, then we have $\rho \geq \frac{\tau / r}{2.5\bar{\kappa}_k^2\log^3 n}$. By combining this with the above result we have
\begin{align*}
\E\|\A(\z_t - \z^*)\|^2 \leq \left(1 - \frac{\tau / r}{10\bar{\kappa}_k^2\log^3 n}\right)^t \cdot \|\A(\z_0 - \z^*)\|^2.
\end{align*}
Notice that this is the convergence result for solving $\A^\top\A\z = \c$ for an arbitrary $\c$. By choosing $\c = \A^\top\b$, we have $\z^* = (\A^\top\A)^\dagger\c = (\A^\top\A)^\dagger\A^\top\b = \A^\dagger\b$ which is the least squares solution of linear system $\A\z = \b$ we want. Finally, notice that both the matrix $\A$ and vector $\c$ here are after the randomized Hadamard transform, and we need to transform the variable back. Since $\Q^\top\Q=\I$, the quantity $\A\z$ here is the identical to the original $\A\x$, thus we have
\begin{align*}
\E\|\A(\x_t - \x^*)\|^2 \leq \left(1 - \frac{\tau / r}{10\bar{\kappa}_k^2\log^3 n}\right)^t \cdot \|\A(\x_0 - \x^*)\|^2.
\end{align*}
By further choosing $\tau = O(k\cdot \log^3 n), \x_0 = \mathbf{0}$ and $t_{\max} = O(r\bar{\kappa}_k^2/k \cdot \log(100/\epsilon))$ we have $\E\|\A(\x_t - \x^*)\|^2 \leq \frac{\epsilon}{100} \|\A\x^*\|^2$. By applying Markov's inequality we have $\|\A(\x_t - \x^*)\|^2 \leq \epsilon \|\A\x^*\|^2$ holds with probability $0.98$. Finally by using Lemma~\ref{lem:precondition_sparse} we compute the total cost as
\begin{align*}
& ~ O(mn\log n) + t_{\max} \cdot O(m\tau \log(\tau/\delta_1) + \tau^{\omega} + m\tau s_{\max}) \\
= & ~ O(mn\log n) + r\bar{\kappa}_k^2/k \cdot \log(100/\epsilon) \cdot O(mk \log^3 n \log (r\bar{\kappa}_k) + k^\omega \log^{3\omega} n) \\
= & ~ O(mn \log n + (mr\log^3 n \log (r\bar{\kappa}_k) + rk^{\omega-1} \log^{3\omega}n)\cdot\bar{\kappa}_k^2\log1/\epsilon) \\
= & ~ O(mn \log n + (mr\log^4 n + rk^{\omega-1} \log^{3\omega}n)\cdot\bar{\kappa}_k^2\log1/\epsilon)
\end{align*}
where the last step follows from the assumption $\bar{\kappa}_k \leq n^4$ and thus $\log (r\bar{\kappa}_k) = O(\log n)$.
\end{proof}

\section{Least Squares Regression}
\label{s:ls}

In this section, consider a tall least squares problem where, given an $n\times d$ matrix $\A$ and $\b\in\R^n$, where $n\gg d$, we want to find 
\begin{align*}
\x^*=\argmin_\x\|\A\x-\b\|^2.    
\end{align*}
Note that one could simply apply our Algorithm~\ref{alg:main_2} and use Theorem~\ref{thm:main_2} to get the complexity analysis, obtaining runtime $\tilde O(nd + dk^{\omega-1})$. However, in this setting, we can further improve on this approach for sparse matrices by combining it with a standard sketch-to-precondition iterative solver for tall least squares \cite{rokhlin2008fast}. In this approach, one first constructs a nearly square sketch $\tilde \A$ of the tall matrix $\A$, by using a subspace embedding such as the one from Lemma \ref{lem:precondition_sparse}, at the cost of $\tilde O(\nnz(\A)+d^2)$. Then, one constructs a preconditioner of $\A$ by performing an SVD (or QR) decomposition of $\tilde\A$, at the cost of $\tilde O(d^\omega)$. Now, we can simply run a preconditioned iterative solver in time $\tilde O(\nnz(\A) + d^2)$ to solve the least squares task. For matrices with $k$ large singular values, the preconditioner construction can be further accelerated by using approximate SVD with power iteration in time $\tilde O(d^2k^{\omega-2})$. 

In this section, we show that instead of directly constructing the preconditioner from $\tilde\A$, we can use our Algorithm \ref{alg:main_2} as an inner solver, which will act as a preconditioner for the outer iterative method. In this case, as our outer iterative method we use preconditioned gradient descent.

\begin{algorithm}[!ht]
\caption{Least squares solver using gradient descent with inexact preconditioning.}
\label{alg:main_ls}
\begin{algorithmic}[1]
\State \textbf{Input: }matrix $\A\in\R^{n\times d}$, vector $\b \in \R^n$, uniform sample size $\tau$, iterate $\x_0$, inner CG iteration $s_{\max}$, inner iteration $t_{\max}$, outer iteration $T$;
\State Compute sketch $\tilde{\A} \leftarrow \mathbf{\Phi}\A$; \Comment{Takes $O(\nnz(\A)\log(d/\delta))$ time.}
\For{$t = 0, 1, \ldots, T-1$}
\State Compute $\g_t \leftarrow \A^\top(\A\x_t - \b)$; \Comment{Takes $O(\nnz(\A))$ time.}
\State Solve $\hat{\p}_t \approx (\tilde{\A}^\top\tilde{\A})^\dagger\g_t$ by calling Algorithm~\ref{alg:main_2} with choice $(\tilde{\A}, \g_t, \tau, \mathbf{0}, t_{\max}, s_{\max})$;
\State Compute $\x_{t+1} \leftarrow \x_t - \hat{\p}_t$;
\EndFor \\
\Return $\tilde{\x} = \x_{t_{\max}}$;
\end{algorithmic}
\end{algorithm}

\begin{theorem}[Least squares]\label{thm:main_ls}
Given matrix $\A\in\R^{n \times d}$, $\b \in \R^n$ with $\rank(\A) = r$, let $\{\sigma_i\}_{i=1}^r$ be the singular values of $\A$ in decreasing order. For $\log d \leq k < r$, define $\bar{\kappa}_k := (\frac{1}{r-k}\sum_{j > k}^r \sigma_j^2 / \sigma_r^2)^{1/2}$ and assume for simplicity that $\bar{\kappa}_k \leq d^4$.
Given $\epsilon > 0$, running Algorithm~\ref{alg:main_ls} with choice $\tau \geq Ck\log^3 d$ for some $C=O(1)$, $\x_0 = \mathbf{0}, s_{\max} = O(\log(r\bar{\kappa}_k / \tau)), t_{\max} = O(r\bar{\kappa}_k^2/k)$ and $T = O(\log(1/\epsilon))$ will output $\tilde \x$ such that  with probability $0.99$ we have $\|\A\tilde\x-\b\|^2\leq\min_\x\|\A\x-\b\|^2 + \epsilon\|\b\|^2$ in time 
\begin{align*}
O\left(\nnz(\A)\log(d/\epsilon) + (d^2\log^4 d + dk^{\omega-1}\log^{3\omega}d)\cdot \bar{\kappa}_k^2\log1/\epsilon\right).
\end{align*}
\end{theorem}
\begin{proof}
Given $\A\in\R^{n \times d}$ with $\rank(\A) = r$ and $\b \in \R^n$, we consider solving the least squares problem $\x^* = \argmin_{\x}\|\A\x-\b\|$ with (approximate) preconditioned gradient descent
\begin{align*}
\x_{t+1} = \x_t - \tilde{\H}^\dagger \g_t = \x_t - \underbrace{(\tilde{\A}^\top\tilde{\A})^\dagger (\A^\top\A\x_t - \A^\top\b)}_{\p_t^*}
\end{align*}
where $\tilde{\H} = \tilde{\A}^\top\tilde{\A} = \A^\top\mathbf{\Phi}^\top\mathbf{\Phi}\A$ is such that $\tilde{\H} \approx_{1+\epsilon_H} \H = \A^\top\A$ for some constant $\epsilon_H$ that will be specified later. By using the fact that $\H(\x_t - \x^*) = \g_t$ we have $\|\g_t\|_{\H^\dagger} = \|\x_t - \x^*\|_{\H}$, and
\begin{align*}
\|\x_t - \p_t^* - \x^*\|_{\H}^2 = & ~ \|\x_t - \x^*\|_{\H}^2 - 2(\x_t - \x^*)^\top\H\p_t^* + \|\p_t^*\|_{\H}^2 \\
= & ~ \|\x_t - \x^*\|_{\H}^2 - 2\g_t^\top\tilde{\H}^\dagger\g_t + \|\H^{1/2}\tilde{\H}^\dagger\g_t\|^2 \\
\leq & ~ \|\x_t - \x^*\|_{\H}^2 - 2\|\g_t\|_{\tilde{\H}^\dagger}^2 + (1+\epsilon_H) \|\g_t\|_{\tilde{\H}^\dagger}^2 \\
= & ~ \|\x_t - \x^*\|_{\H}^2 - (1-\epsilon_H)\|\g_t\|_{\tilde{\H}^\dagger}^2 \\
\leq & ~ \|\x_t - \x^*\|_{\H}^2 - (1-\epsilon_H)^2\|\g_t\|_{\H^\dagger}^2 \\
\leq & ~ 2 \epsilon_H\cdot \|\x_t - \x^*\|_{\H}^2
\end{align*}
By choosing $\epsilon_H = 1/ 8$ we have $\|\x_t - \p_t^* - \x^*\|_{\H} \leq \frac{1}{2}\|\x_t-\x^*\|_{\H}$. Suppose for now that we can construct a randomized estimate $\hat{\p}_t$ such that $\E_t[\|\hat{\p}_t - \p_t^*\|_{\tilde{\H}}] \leq \frac{1}{4} \|\p_t^*\|_{\tilde{\H}}$, where $\E_t$ is the expectation conditioned on $\x_t$ and $\tilde\H$. Then, since we also assume that $\tilde{\H} \approx_{1+\epsilon_H} \H$ for $\epsilon_H = 1/8$, we have
\begin{align*}
\E_t\big[\|\hat{\p}_t - \p_t^*\|_{\H}\big] \leq \sqrt{1+\epsilon_H} \E_t\big[\|\hat{\p}_t - \p_t^*\|_{\tilde{\H}}\big] \leq \frac{\sqrt{1+\epsilon_H}}{4} \|\p_t^*\|_{\tilde{\H}} \leq \frac{1+\epsilon_H}{4} \|\p_t^*\|_{\H} \leq \frac{1}{3} \|\p_t^*\|_{\H}.
\end{align*}
By combining the above two results we have
\begin{align*}
\E_t\big[\|\A(\x_{t+1} - \x^*)\|\big]
= & ~ \E_t\big[\|\x_t - \hat{\p}_t - \x^*\|_{\H}\big] \\
\leq & ~ \|\x_t - \x^* - \p_t^*\|_{\H} + \E_t\big[\|\p_t^* - \hat{\p}_t\|_{\H}\big] \\
\leq & ~ \frac{1}{2}\|\x_t -\x^*\|_{\H} + \frac{1}{3} \|\p_t^*\|_{\H} \\
= & ~ \frac{1}{2}\|\x_t -\x^*\|_{\H} + \frac{1}{3} \|\g_t\|_{\H^\dagger} \\
= & ~ \frac{5}{6} \|\A(\x_t - \x^*)\|.
\end{align*}
To achieve this convergence result, there are two parts left that need to be done: (i) construct the preconditioner $\tilde{\H} = \tilde{\A}^\top\tilde{\A}$ such that $\tilde{\H} \approx_{1+1/8} \H$ is a spectral approximation; and (ii) compute the randomized estimate $\hat{\p}_t$ such that $\E_t\big[\|\hat{\p}_t - \p_t^*\|_{\tilde{\H}}\big] \leq \frac{1}{4} \|\p_t^*\|_{\tilde{\H}}$.
For the first part, we use sparse embedding matrix to construct $\tilde{\A} = \mathbf{\Phi}\A$. According to Lemma~\ref{lem:precondition_sparse} we know $\tilde{\A}^\top\tilde{\A} \approx_{1+1/8} \A^\top\A$ holds with probability at least $1-\delta_1$ by choosing $\phi = O(d+\log(1/\delta_1))$, and the construction takes $O(\nnz(\A)\log(d/\delta_1) + d^2\log^4 (d/\delta_1))$ time. 

For the second part, to approximate $\p_t^*$, we treat it as solving the linear system $(\tilde{\A}^\top\tilde{\A}) \p = \g_t$ where $\tilde{\A} \in \R^{\phi \times d}$ for $\phi = O(d+ \log (1 / \delta_1))$ and $\g_t = \A^\top(\A\x_t - \b)$ is given. Notice that since $\tilde{\A}^\top\tilde{\A}$ is an $(1+\frac{1}{8})$-spectral approximation of $\A^\top\A$, we know that all the singular values of $\tilde{\A}$ are constant factor approximations of the singular values of $\A$, thus we know $\bar{\kappa}_k(\tilde{\A})$ and $\bar{\kappa}_k(\A)$ only differs within a constant factor. We apply Algorithm~\ref{alg:main_2} to $(\tilde{\A}\in\R^{\phi\times d}, \c = \g_t, \p_0 = \mathbf{0})$ with choice $\tau \geq Ck\log^3 d, s_{\max} = O(\log(r\bar{\kappa}_k/k))$ and $t_{\max} = O(r\bar{\kappa}_k^2/k)$ (note that the RHT preprocessing step from line 2 can use the same RHT across all gradient descent iterations). According to the proof of Theorem~\ref{thm:main_2}, we know that with the right choice of constants the output $\hat{\p}_t$ will satisfy $\E_t[\|\hat{\p}_t - \p_t^*\|_{\tilde{\H}}] \leq \frac{1}{4} \|\p_t^*\|_{\tilde{\H}}$ in time
\begin{align*}
O(\phi d \log d + (\phi r \log^4 d + rk^{\omega-1}\log^{3\omega}d)\bar{\kappa}_k^2) = O((d^2\log^4 d + dk^{\omega-1}\log^{3\omega}d)\cdot \bar{\kappa}_k^2).
\end{align*}
With the above analysis,  we have the following convergence result, where the expectation is taken over the entire run of the algorithm, conditioning on successful construction and preprocessing of $\tilde\A$ (which holds with high probability):
\begin{align*}
 \E\big[\|\A(\x_T-\x^*)\|\big]\leq \Big(\frac 56\Big)^{T}\|\A(\x_0-\x^*)\|
\end{align*}
Thus by choosing $\x_0=\mathbf 0$, $T = O(\log(1/\epsilon))$, using Markov's inequality we can obtain that with probability 0.99, $\tilde{\x} = \x_T$ will satisfy $\|\A(\tilde\x-\x^*)\|^2\leq\epsilon^2\|\b\|^2$, which in turn implies $\|\A\tilde{\x} - \b\|^2 \leq \|\A\x^* - \b\|^2+\epsilon\|\b\|^2$.
The overall time complexity of the procedure is:
\begin{align*}
& ~ O\left(\nnz(\A)\log d + d^2\log^4 d + (\nnz(\A) + \bar{\kappa}_k^2 d^2\log^4 d + \bar{\kappa}_k^2 dk^{\omega-1}\log^{3\omega}d)\log1/\epsilon\right)\\
= & ~ O\left(\nnz(\A)\log(d/\epsilon) + (d^2\log^4 d + dk^{\omega-1}\log^{3\omega}d)\cdot \bar{\kappa}_k^2\log1/\epsilon\right).
\end{align*}

\end{proof}

\bibliographystyle{alpha}
\bibliography{ref.bib}

\appendix
\section{Omitted Proofs}\label{sec:appendix_proof}

\begin{proof}[Proof of Lemma~\ref{lem:sample_log}]
We first state Lemma~\ref{lem:geo_tail} which is the tail bounds for sums of geometric random variables.
\begin{lemma}[Theorem 2.1 in \cite{j18}]\label{lem:geo_tail}
Let $\{X_i\}_{i=1}^n$ be independent geometric random variables with possibly different distributions: $X_i \sim p_i$ where $0 < p_i \leq 1$. Let $X := \sum_{i=1}^n X_i, \mu := \E[X]$ and $p_* := \min_i p_i$. For any $\lambda \geq 1$, we have
\begin{align*}
\Pr\{X \geq \lambda \mu\} \leq e^{-p_*\mu(\lambda-1-\ln\lambda)}.
\end{align*}
\end{lemma}
For $i\in[t]$, we define random variable $X_i$ as the number of calls it take to obtain a new sample in the $i$-th step (that is, we have obtained $i-1$ samples and want to get the $i$-th new sample). Then we know that $\{X_i\}_{i=1}^t$ are  independent geometric random variables with different distributions: $X_i \sim \mathrm{Ge}(p_i)$ where $p_i = 1- \frac{k+i-1}{m}$. We denote $X := \sum_{i=1}^t X_i$ and also define
\begin{align*}
\begin{cases}
\mu :=  \E[X] = \sum_{i=1}^t \E[X_i] = \sum_{i=1}^t \frac{1}{p_i} \\
p_* := \min_i p_i = 1 - \frac{k+t-1}{m}
\end{cases}
\end{align*}
Then we have
\begin{align*}
\mu = & ~ \frac{m}{m-k} + \frac{m}{m-k-1} + \cdots + \frac{m}{m - k - t +1} \\
= & ~ \frac{1}{1 - \frac{k}{m}} + \frac{1}{1 - \frac{k+1}{m}} + \cdots + \frac{1}{1 - \frac{k+t-1}{m}} \\
\leq & ~ (1 + \frac{2k}{m}) + (1 + \frac{2(k+1)}{m}) + \cdots + (1 + \frac{2(k+t-1)}{m}) \\
= & ~ t + \frac{2kt + t(t-1)}{m} \\
= & ~ t \cdot \left(1+\frac{2k+t-1}{m}\right)
\end{align*}
where the third step follows from $k, t \ll m$. Next we lower bound $p_* \mu$:
\begin{align*}
p_* \mu = & ~ \left(1 - \frac{k+t-1}{m}\right)\left(\frac{1}{1 - \frac{k}{m}} + \frac{1}{1 - \frac{k+1}{m}} + \cdots + \frac{1}{1 - \frac{k+t-1}{m}}\right) \\
= & ~ t - \sum_{i=1}^t \frac{t- i}{m-(k+i-1)} \\
\geq & ~ t - t\cdot \frac{t-1}{m-k} \\
\geq & ~ \frac{2}{3} t
\end{align*}
where the third step follows since $\frac{t- i}{m-(k+i-1)}$ is decreasing in $i$ and the last step follows from $k, t \ll m$. Now we can bound the tail probability by using concentration inequality of geometric distribution. We use Lemma~\ref{lem:geo_tail} by setting $\lambda = 2 + \frac{2}{t}\log\frac{1}{\delta}$ and have
\begin{align*}
\Pr\left\{X \geq (2 + \frac{2}{t}\log\frac{1}{\delta}) \mu\right\} \leq & ~  \exp(-p_*\mu(1+\frac{2}{t}\log\frac{1}{\delta}-\log(2 + \frac{2}{t}\log\frac{1}{\delta}))) \\
\leq & ~ \exp(-\frac{2t}{3}(1+\frac{2}{t}\log\frac{1}{\delta}-\log(2 + \frac{2}{t}\log\frac{1}{\delta})))
\end{align*}
In order to show that RHS is less than $\delta$, we only need to show $1+\frac{2}{t}\log\frac{1}{\delta}-\log(2 + \frac{2}{t}\log\frac{1}{\delta}) \geq \frac{3}{2t} \log\frac{1}{\delta}$. We define $f(x) = 1 + \frac{x}{2} - \log(2+2x)$, then we know $f(x)$ is decreasing for $0 \leq x \leq 1$ and increasing for $x >1$, thus $f_{\min} = f(1) = \frac{3}{2} - \log 4 > 0$. Since $\frac{1}{t} \log\frac{1}{\delta} > 0$ we have $f(\frac{1}{t} \log\frac{1}{\delta}) > 0$, which gives
\begin{align*}
\Pr\left\{X \geq (2 + \frac{2}{t}\log\frac{1}{\delta}) \mu\right\} \leq \delta.
\end{align*}
This shows that with probability at least $1-\delta$, the number of calls to the oracle is at most
\begin{align*}
X \leq (2 + \frac{2}{t}\log\frac{1}{\delta})\cdot t \cdot \left(1+\frac{2k+t-1}{m}\right) \leq 4t + 4 \log\frac{1}{\delta}
\end{align*}
which is $O(t+\log\frac{1}{\delta})$.
\end{proof}

\end{document}